\documentclass[journal,twosided,web]{ieeecolor}

\usepackage{generic}
\usepackage{cite}
\usepackage{amsmath,amssymb,amsfonts}
\usepackage{algorithmic}
\usepackage{graphicx}
\usepackage{textcomp}
\def\BibTeX{{\rm B\kern-.05em{\sc i\kern-.025em b}\kern-.08em
    T\kern-.1667em\lower.7ex\hbox{E}\kern-.125emX}}
\usepackage{cases}
\usepackage{dsfont}

\newcommand{\A}{{\mathcal{A}}}

\newcommand{\bbP}{P}

\newcommand{\vs}{\vspace{-1mm}}
\newcommand{\hs}[1][0.5]{\hspace{-#1mm}}
\DeclareMathOperator*{\argmax}{arg\,max}

\newtheorem{definition}{Definition}
\newtheorem{theorem}{Theorem}
\newtheorem{proposition}{Proposition}
\newtheorem{lemma}{Lemma}

\newcommand{\online}[1]{This proof is omitted for space and appears online \cite{collins2021robust}.}
\renewcommand{\online}[1]{#1}

\begin{document}
\title{A Coupling Approach to Analyzing Games with Dynamic Environments}
\author{Brandon C. Collins, \IEEEmembership{Student Member, IEEE}, Shouhuai Xu, \IEEEmembership{Senior Member, IEEE} and Philip N. Brown, \IEEEmembership{Member, IEEE}
\thanks{A preliminary version of this work appeared in~\cite{collins2021cdc}.}
\thanks{This work was supported in part by NSF Grants \#2122631, \#2115134, DEB-\#2032465, and ECCS-\#2013779, ARO Grant \#W911NF-17-1-0566, and Colorado State Bill 18-086.}
\thanks{ Brandon C. Collins, Shouhuai Xu, and Philip N. Brown are with the Department of Computer Science, University of Colorado Colorado Springs, Colorado Springs, CO 80918 USA (emails: \{bcollin3,sxu,pbrown2\}@uccs.edu).}
}

\maketitle

\begin{abstract}
The theory of learning in games has extensively studied situations where agents respond dynamically to each other by optimizing a fixed utility function.
However, in real situations, the strategic environment varies
as a result of past 
agent
choices.
Unfortunately, the analysis techniques that enabled  a rich characterization of the emergent behavior in 
static environment games fail 
to cope with dynamic environment games.
To address this, we develop a general framework using probabilistic couplings 
to extend the analysis of static environment games to dynamic ones.
Using this approach, we obtain sufficient conditions under which traditional characterizations of Nash equilibria with best response dynamics and stochastic stability with log-linear learning can be extended to dynamic environment games.
As a case study, we pose a model of cyber threat intelligence sharing between firms and a simple dynamic game-theoretic model of social precautions in an epidemic, both of which feature dynamic environments.
For both examples, we obtain conditions under which the emergent behavior is characterized in the dynamic game by performing the traditional analysis on a reference static environment game.

\end{abstract}


\section{Introduction}
\label{sec:introduction}
In social systems and distributed engineering systems, collective behavior is the result of many individuals making intertwined self-interested choices.
    In many cases, the value of a particular choice depends not only on the current choices being made by others, but also on the history of past choices.

	In principle, these socio-environmental feedback loops can be analyzed using techniques from \emph{game theory}, which has a long history of analyzing the society-scale effects of self-interested behaviors.
	For instance, game theory has long been used to study the spread of social conventions \cite{young1993evolution} using models such as the graphical coordination game \cite{kearns2013graphical} with the stochastic learning algorithm known as log-linear learning \cite{marden2012revisiting}.
    However, traditional analysis techniques almost uniformly assume that the game's utility functions are fixed for all time, so that the agents' choices over time can be described by a stationary Markov process.
	However, such analyses fail or become unwieldy when utility functions depend on the history of play. 
	
	Analysis techniques for history-dependent games have many applications. 
	For example, in a global pandemic, the individual choice to adopt protective measures (e.g., wearing masks) may be made in response to the behavior of others and the prevalence of the disease.
	In turn, the prevalence of the disease is a function of the history of individual choices to adopt protective measures.
	As another example, game theoretic methods are frequently proposed in the area of distributed control of multiagent systems~\cite{Chandan2019,Collins2020,marden2013distributed,kanakia2016modeling}.
	However, in a distributed control application, agents' actions may directly modify the strategic environment; for instance if a search-and-rescue UAV identifies a disaster victim, that victim may be removed from the list of other UAVs' objectives.
    Other applications that can be modeled by history-dependent games are in machine learning \cite{wang2019evolutionary,garciarena2018evolved,costa2019coegan}
    and biology \cite{tilman2017maintaining,tilman2020evolutionary}.


Dynamic environments may be either random or deterministic functions of agent behaviors.
    The random case has been used to understand cell behavior and evolution \cite{wolf2005diversity},
    where the environment is a Markov model whose state
    represents the current conditions the cell inhabits.
    This approach is further studied in \cite{mehta2015mutation}, which characterizes the importance of genetic mutations in dynamic environments.
    On the other hand, the deterministic case is often studied in the context of dynamical systems, where differential equations dictate the evolution of the environment.
    One such example is \cite{weitz2016oscillating}, where the authors characterize an oscillating tragedy of the commons effect under certain environmental feedback scenarios.
    Revisiting population games, \cite{mai2018cycles} studies a zero-sum Rock Paper Scissors style game with replicator dynamics, where the environment responds negatively to the increase in frequency of any individual population.
    They show that the model is recurrent over time.
    This result is extended in \cite{skoulakis2020evolutionary} where the recurrence is generalized to a class of dynamics with environments and populations interacting in a networked fashion.
    Particularly, it is shown that the class of dynamics is equivalent to a zero-sum poly-matrix game under replicator dynamics and can be studied using traditional game theoretic techniques.
    
    In this paper, we develop a general framework for analyzing binary-action games with dynamic environments, which we term {\em history-dependent games}.
    We develop a probabilistic coupling between a reference static game and a history-dependent game.
    We show that if the utility functions of the history-dependent game can be referenced appropriately to the static game, then traditional game theoretic results on the static game can be extended to the history-dependent game.
    Specifically, the coupling provides a general inequality (Lemma~\ref{thm:proof of stoch dom+}) which 
    compares the probability that the history-dependent game is in a given state to that of the static game, for all time steps.
    Using this inequality, we develop sufficient conditions under which traditional characterizations of Nash equilibria with best response dynamics and stochastic stability with log-linear learning can be extended to history-dependent games.
    
    To show the applicability of our results, we consider two examples of history-dependent games in the contexts of cybersecurity and epidemiology.
    The first example (given in Section~\ref{sec:CTI}) is a Cyber Threat Intelligence (CTI) sharing game where firms generate CTI over time and decide whether to share CTI with each other.
    This problem has been studied widely using game theory \cite{hausken2007,gao2014,gao2016differential,DBLP:conf/icc/ToshSKKM15,solak2020optimal,collins2021paying}, but all of these works consider a fixed value of CTI.
    This is an unrealistic assumption, because different pieces of CTI may have different intrinsic values or due to timeliness \cite{wagner2019cyber}.
    We formulate a CTI sharing game which admits a broad class of variable CTI functions (encoding the environment) and show sufficient conditions under which the all-sharing action profile
    is comparable to a Nash equilibrium in a fixed setting with best response dynamics and is stochastically stable with log-linear learning.
    
    In the second example (developed in Section~\ref{sec:epidemics}), we study the feedback between the prevalence of an epidemic and individuals' adoption of preventative measures.
    It is known that preventative measures like mask wearing or social distancing reduces the risk of infection spreading between individuals.
    However, an individual's willingness to practice such measures depends on prevalence of the epidemic.
    That is, individuals are more likely to adopt preventative measures when the epidemic is widespread, which then in turn mitigates the epidemic, creating a feedback loop.
    To capture this, we intertwine the graphical coordination game model of social conventions and the compartmental epidemic SIS model.
    In this new model, the coordination game's payoffs are impacted by the current state of the SIS model, and the infectiousness parameter of the SIS model depends on the current choices of agents.
    We use the 
    coupling to derive sufficient conditions under which all individuals practice preventative conventions as the unique stochastically stable state of log-linear learning.

\section{Model}
    \subsection{Game Formulation}
        In this work, we consider binary action games.
        Let $N=\{1,2,3,\dots,n\}$ denote the set of agents; agent $i\in N$ has action set $A_i=\{0,1\}$.
        The joint action space is then given by $A=\{0,1\}^{|N|}$.
    	We denote an action profile as $a\in A$ and use $a_i$ to denote agent $i$'s action.
    	Throughout, we use state and action profile interchangeably to refer to a vector of agent choices $a$.
    	We denote actions of all other agents by $a_{-i}=(a_1,a_2,\dots,a_{i-1},a_{i+1},\dots,a_{|N|})$.
    	We denote the all-$1$ action profile as $\vec{1}=(1)^{|N|}_{i=1}$ and similarly for the all-zero profile $\vec{0}$.
    	Further, let $\Delta(A)$ denote the standard probability simplex over $A$.
    	Let $U_i:A\rightarrow \mathbb{R}$ be agent $i$'s utility function, and let $\mathcal{U}$ denote the space of all utility functions $U_i:A\rightarrow\mathbb{R}$.
    	We write $U=\{U_i\}_{i\in N}$
    	to denote the collection of utility functions for all agents.
    	Thus, we specify a static game using tuple $g=(N,A,U)$.
    	

	    We consider a general class of games which generalize the preceding static game to enable the past behavior of agents to influence the current behavior.
	    To accomplish this, we allow the utility function to be a function of previous actions taken by agents.
	    We begin by formalizing the past behavior or history of agent actions.
	    %
	    We use $\A_T$ to denote the set of joint action histories of length $T\in\mathbb{N}$,
	    and denote the set of all histories as $\mathcal{A}=\cup_{T\in \mathbb{N}}\mathcal{A}_T$.
	    We use $\alpha\in\mathcal{A}_T$ to refer to a history of action profiles (which we refer to as a \emph{path}) and use superscripts to denote time indices so that $\alpha=(\alpha^1,\dots,\alpha^T)$.
	    We 
	    use $\alpha^T$ to denote the last action profile on a path $\alpha$ of length $T$.
	    We define $A$ and $\mathcal{A}_T$ as partially ordered sets with
	    partial order $\geq_A$, where $a'\geq_A a$ when $a,a'\in A$ and $a'_i\geq a_i$ for all $i\in N$, recalling that $a'_i,a_i\in \{0,1\}$.
	    Using this, we further define partial order $\geq_{\mathcal{A}_T}$, such that $\bar{\alpha}\geq_{\mathcal{A}_T}\alpha$ when $\alpha,\bar{\alpha}\in \mathcal{A}_T$ and $\bar{\alpha}^t\geq_{A}\alpha^t$ for all $t\in\{1,2,...,T\}$.
	    
	    To model history-dependent utility functions, let $U^\alpha_i:A\rightarrow \mathbb{R}$, where this utility function is not only specific to agent $i$ but also to the history $\alpha$.
    	Let $U^\alpha=(U^\alpha_1,U^\alpha_2,...,U^\alpha_{|N|})$ denote each agent's utility function given history $\alpha$
    	and let $U^\mathcal{A}=\{U^\alpha\mid\alpha\in\mathcal{A}\}$ be the set of utility functions across all paths.
        We denote a \emph{history-dependent game} as tuple $(N,A,U^{\mathcal{A}})$ and let $\mathcal{G}^{\mathcal{A}}$ be the set of all such tuples.
        
        In general, the analysis of history-dependent games is difficult as the utility function can vary wildly between each time step.
        Accordingly, we restrict our attention to a subset of history-dependent games that have an important utility function property, namely that there exists a reference static game with certain properties relating the history-dependent and static utility function.
        This is the basis for us to extend the traditional game theoretic analysis of static games to history-dependent games.
        We call this specific class {\em aligned history-dependent games}, formally defined as follows. 
        

    	\begin{definition}[aligned history-dependent game]\label{def:binary action history var}
    		We call a tuple $g=(N,A,U^\mathcal{A})\in \mathcal{G}^{\mathcal{A}}$ an \textit{aligned history-dependent game} if there exists a static game $\hat{g}=(N,A,\hat{U})$:
    		\begin{enumerate}
    			\item $U^\alpha_i(1,\alpha^T_{-i})\geq\hat{U}_i(1,a_{-i})$
    			\item $\hat{U_i}(0,a_{-i})\geq U^{\alpha}_i(0,\alpha^T_{-i})$
    		\end{enumerate}
            for any $\alpha\in\mathcal{A}$, $a,a'\in A$, $T\in\mathbb{N}$ such that $\alpha^T_{-i}\geq_{A_{-i}} a_{-i}$ and $a,a'$ vary by only a unilateral deviation.         
            For convenience, we denote ordering $\geq_{A_{-i}}$ over $A_{-i}=\{0,1\}^{|N|-1}$ equivalently to $\geq_A$.
        \end{definition}
    
    Since we focus on analyzing history-dependent games corresponding to
    relatively well-understood static ones, Definition \ref{def:binary action history var} ensures that a suitable static game exists for comparison. 
    Intuitively, Definition \ref{def:binary action history var} defines a subclass of history-dependent games such that the 1 actions benefit from history relative to a reference static game, but the 0 actions can only lose due to history.
    More specifically, it ensures that playing the $1$ action is always more desirable for agents in the history-dependent game relative to the static game whenever $\alpha^T_{-i}\geq_{A_{-i}}a_{-i}$ (i.e., when a superset of agents are playing $1$ in the history-dependent game relative to the static game).
    Thus, an aligned history-dependent game has the property that for all histories, having more agents playing $1$ can only make playing $1$ more desirable for other agents. 
    

\subsection{Properties of Learning Rules in Games}
    To establish maximum generality for our results, we provide conditions on learning rules under which the couplings can be applied, and show that both best response and log-linear learning satisfy these conditions.
    We begin by defining an individual learning rule as a function which gives the probability that action profile $a$ will deviate to $a'$ in a single time step given that only agent $i$ can update its action.
    \begin{definition}[individual learning rule]\label{def:individual}
        Function $P_i:A\times A\times  \mathcal{U}^n\rightarrow [0,1]$ is an individual learning rule if for any $a,a'\in A$, $U\in\mathcal{U}^n$ we have
        \begin{enumerate}
            \item $\sum_{\bar{a}\in A}P_i(a,\bar{a},U)=1$,
            \item $P_i(a,a',U)\geq 0$, and
            \item $P_i(a,a',U)=0$ $a_{-i}\neq a'_{-i}$.
        \end{enumerate}
    \end{definition}
        The first two conditions 
        say that $P_i$ is a valid probability measure over $A$ given prior action $a$ and utility vector $U$.
        The third condition ensures that only agent $i$ changes its action with a positive probability.
        We now give a condition on individual learning rules to derive our main results.
    \begin{definition}[local individual learning rule]
        We say an individual learning rule $P_i$ is \emph{local} if there exists $\bar{P}_i: A_i\times \mathbb{R}^{|A_i|}\rightarrow [0,1]$ such that for any $a\in A$, $a'_i\in A_i$ we have:
        \begin{equation}\label{eq:local LR}
            P_i(a,(a'_i,a_{-i}),U)=\bar{P}_i(a'_i,\vec{U}_i(a_{-i}))
        \end{equation}
where $\vec{U}_i(a_{-i})=(U_i(a_i,a_{-i}))_{a_i\in A_i}$ is a vector of payoffs for each possible action of agent $i$ given the actions of other agents $a_{-i}$.
    \end{definition}
    
    Learning rule $P_i$ being local has several implications.
    First, the probability that agent $i$ selects a given action only depends on its utility function $U_i$, which is a property known as {\em uncoupled} \cite{hart2003uncoupled}.
    Additionally, the probability does not depend on agent $i$'s previous action $a_i$ (i.e., agent $i$ cannot be biased toward their previous action in any way).
    We now pose an additional monotonicity condition on learning rules to ensure that increases in payoffs do not decrease the probability that an action is played.
\begin{definition}[monotone individual learning rule]
    An individual learning rule $P_i:A\times A\times  \mathcal{U}^n\rightarrow [0,1]$ is \emph{monotone with respect to utility} if for any utility function vector $U\in\mathcal{U}^n$, individual agent's action $a'_i\in A_i$, action profile $a\in A$,  nonnegative constant $l\geq0$, and defining $\bar{U}=(U_1,U_2,\dots,U_{i-1},\bar{U}_i,U_{i+1},\dots, U_n)$ and 
    \begin{equation}\label{eq:monotone wrt utility}
        \bar{U}_i(a)=
        \begin{cases}
            U_i(a)+l & a_i=a'_i \\
            U_i(a) & \mbox{else},
        \end{cases}
    \end{equation}
    we have $P_i(a,a',\bar{U})\geq P_i(a,a',U)$ where $a'=(a'_i,a_{-i})$. 
\end{definition}

The preceding condition can be interpreted as follows.
First, select some action of agent $i$, action $a_i'$, and increase its utility such that the increase does not depend on the actions of other agents.
Then, if for any increase in utility the probability agent $i$ selects $a'_i$ does not decrease, the learning rule is monotone with respect to utility.
We now use our definition of individual learning rule to define a learning rule where all agents can update their actions. 

    \begin{definition}[asynchronous learning rule]
    Given a vector of individual learning rules $\vec{P}=(P_1,P_2,\dots,P_n)$, we define an {\em asynchronous learning rule} as
    \begin{equation}\label{eq:asynchronus learning rule}
        P(a,a',U)=\frac{1}{n}\sum_{i\in N}P_i(a,a',U).
    \end{equation}         
    \end{definition}

    We call $P$ asynchronous because it has the property that $P(a,a',u)=0 $ if $a,a'$ vary by more than one action, which follows by Definition~\ref{def:individual}.
    That is, $P$ only permits one agent to change its action at a time.
    More specifically, $P$ selects a single agent $i$ according to a uniform distribution and updates its action according to the distribution given by $P_i$.
    We say asynchronous learning rule $P$ is local and monotone if each individual learning rule in $\vec{P}$ is local and monotone with respect to utility.

Throughout the paper, we couple a learning rule $P$ with game $g$.
We adopt the convention that if game $g\in \mathcal{G}^{\mathcal{A}}$ is a history-dependent game, then the associated learning rules are given by 
\begin{equation} \label{eq:one step with hitory}
    P^\alpha(a'):=P(\alpha^T,a',U^\alpha).
\end{equation}
Additionally, if $\hat{g}=(N,A,\hat{U})$ is a static game, then we similarly define
\begin{equation}\label{eq:one step without hitory}
    \hat{P}^a(a'):=P(a,a',\hat{U}).
\end{equation}
Correspondingly, the individual learning rules are given by
\begin{equation}
\label{eq:individual one step with history}
    P_i^\alpha(a_i'):=P(\alpha^T,(a'_i,\alpha^T_{-i}),U^\alpha)\mbox{, and}
\end{equation}
\begin{equation}\label{eq:individual one step without history}
    \hat{P}_i^a(a'_i):=P(a,(a'_i,a_{-i}),\hat{U}).
\end{equation}

\subsection{Example Learning Rules}
    In the previous section, we defined several properties of learning rules of the form $P:A\times A\times \mathcal{U}\rightarrow [0,1]$.
    In this section, we formulate a variety of well studied learning rules with respect to our learning rule definitions and show they satisfy all of the above properties.

The first example is called log-linear learning, where each individual learning rule $P_i$ is defined as 
        \begin{equation}\label{eq:log-linear learning}
            P_i(a,a',U;\tau)=
            \begin{cases}
             \frac{\exp({\frac{1}{\tau}U_i(a')})}{\sum_{\bar{a}_i\in A_i}\exp({\frac{1}{\tau}U_i(\bar{a}_i,a_{-i})})} & a_{-i}=a'_{-i}\\
                0 & a_{-i}\neq a'_{-i},
            \end{cases}
        \end{equation}
    	where $\exp(x):=e^x$, and $\tau$ is the \emph{temperature} parameter that governs the rationality of agents.
    	As $\tau\to0$, agents best respond with high probability; and as $\tau\to\infty$, agents choose actions uniformly at random.
        Traditional log-linear learning \cite{marden2012revisiting} can be implemented in the framework of \eqref{eq:asynchronus learning rule} by {selecting $\tau$ and letting $\vec{P}=(P_i(\cdot\ ;\tau))^n_{i=1}$.}
        Then, $P(\cdot;\vec{P})$ is equivalent to the previously studied log-linear learning function.
        
        One of the appeals of log-linear learning is that for a special class of static games known as {\em potential games}, log-linear learning has the desirable equilibrium selection properties.
        Formally, a static game $g=(N,A,U)$ is an \emph{exact potential game} if there exists a potential function $\phi$ such that
    	\begin{equation}\label{eq:potential function definition}
        	U_i(a_i',a_{-1})-U_i(a_i,a_{-i})=\phi(a_i',a_{-1})-\phi(a_i,a_{-i})
	    \end{equation}
	    for any $a\in A$, and $a_i,a_i'\in A_i$.
    	Under log-linear learning in potential games, maximizers of the potential function are \emph{stochastically stable} \cite{alos2010logit}.
    	We say $a\in A$ is \emph{strictly} stochastically stable if the following 
    	holds \cite{brown2019}:
    	For any $\epsilon>0$ there exists $\mathcal{T}>0,T<\infty$ such that
		\begin{equation}\label{eqn:SS to g hat}
			\mbox{Pr}(s(t;P)=\vec{1})>1-\epsilon \mbox{ whenever }t>T,\tau<\mathcal{T}
		\end{equation}
	    where $s(\cdot)$ is a random variable representing the action profile at time $t$ under log-linear learning, given temperature $\tau$, initial distribution $\pi\in\Delta(A)$, and game $g$.

	    Exact potential games under log-linear learning may be analyzed using a theory of \emph{resistance trees} \cite{young1993evolution,alos2010logit,pradelski2012learning,marden2012revisiting} to relate potential function maximizers to stochastic stability.
	    However, this analysis depends on the fact that log-linear learning induces an ergodic Markov process on the action profiles,
	    and 
	    it is unclear how to apply resistance tree techniques generally on history-dependent games to show stochastic stability.
        (In Theorem~\ref{thm:SS} we will show how the aligned history-dependent framework may be used to apply stochastic stability to history-dependent games.)
        
        The second example learning rule is the well-studied best response dynamics \cite{young2004strategic}.
        In this learning rule agents always best respond to the actions of other agents (breaking ties uniformly randomly).
        An individual agent's best response learning rule can be given by
        \begin{equation}\label{eq:best response dynamics}
            P_i(a,a',U)=
            \begin{cases}
                \frac{\mathds{1}(a'_i\in \textrm{BR}_i(a_{-i},U))}{|\textrm{BR}_i(a_{-i},U)|} & a_{-i}=a'_{-i}\\
                0 & a_{-i}=a'_{-i}
            \end{cases}
        \end{equation}
	    where $\textrm{BR}_i(a_{-i},U)$ is agent $i$'s best response set given utility function $U$ to the other agent's actions $a_{-i}$.
	    The best response is implemented using \eqref{eq:asynchronus learning rule} by letting $\vec{P}$ be a uniform vector of the above $P_i$ in the form of $P(\cdot:\vec{P})$.
	    
	    The best response learning rule has the property that once it selects a strict Nash equilibrium, the process stays there for all following time steps.
	    Formally, suppose $a\in A$ is a strict Nash equilibrium, then
	    \begin{equation}\label{eq:cournot limit}
	        \Pr(s(t+1,\hat{P}_\pi)=a)\geq\Pr(s(t,\hat{P}_\pi)=a)
	    \end{equation}
	    for all times $t$ and initial distributions $\pi$.

To show that the preceding two example learning rules are applicable to
the aligned history-dependent game framework we propose in this paper, we must show that \eqref{eq:log-linear learning} and \eqref{eq:best response dynamics} are individual, monotone with respect to utility, and local.
        \begin{lemma}\label{thm:individual, monotone, local}
           Log-linear learning given by \eqref{eq:log-linear learning} and best response dynamics given by \eqref{eq:best response dynamics} are individual, monotone with respect to utility, and local.
        \end{lemma}
\begin{proof}
We verify both log-linear learning and best response dynamics are individual, monotone with respect to utility, and local in order.
%
%
Beginning with individual learning rule, we verify that log-linear learning, given by $\eqref{eq:log-linear learning}$, satisfies the three conditions.
            Let $U\in\mathcal{U}^n$ be a vector of utility functions.
            The first condition of individuality is easy to verify algebraically because for some $a\in A$ and some agent $i$
            \begin{equation}
                \sum_{a'\in A}P_i(a,a',U;\tau)=
                \sum_{a'_i\in A_i}\frac{\exp({\frac{1}{\tau}U_i(a'_i,a_{-i})})}{\sum_{\bar{a}_i\in A_i}\exp({\frac{1}{\tau}U_i(\bar{a}_i,a_{-i})})}=1
            \end{equation}
        where the first equality follows as $P(a,a',U)=0$ when $a,a'$ vary by more than a singe agent's action.
        A similar argument holds for the best response dynamics so we omit it for space.
        It is easy to see the second condition of Definition~\ref{def:individual} for both learning rules as both are strictly nonnegative by definition.
        The third condition holds since both learning rules explicitly have $P_i(a,a')=0$ whenever $a_{-i}\neq a'_{-i}$.
        
        We now show both learning rules are monotone with respect to utility.
        Beginning with log-linear learning, define $\bar{U}_i(a)$ according to \eqref{eq:monotone wrt utility} for some agent $i$, $l\geq0$, and action $a_i'\in A_i$.
        Let $\bar{U}=(U_1,U_2,\dots,\bar{U}_i,\dots,U_n)$ and let $a'=(a'_i,a_{-i})$; we show $P(a,a',\bar{U})\geq P(a,a',U)$ directly:
        \begin{equation}
        \begin{aligned}
            P(a,a',\bar{U})&=\frac{\exp(\frac{1}{\tau}l)\exp({\frac{1}{\tau}U_i(a')})}{\sum_{\bar{a}_i\in A_i}\exp({\frac{1}{\tau}\bar{U}_i(\bar{a}_i,a_{-i})})}\\
            &\geq\frac{\exp(\frac{1}{\tau}l)\exp({\frac{1}{\tau}U_i(a')})}{\exp(\frac{1}{\tau}l)\sum_{\bar{a}_i\in A_i}\exp({\frac{1}{\tau}U_i(\bar{a}_i,a_{-i})})}\\
            &=P(a,a',U)
        \end{aligned}
        \end{equation}
        for any action $a$. 
        
    	Monotonicity with respect to utility can be seen for best response dynamics in two cases.
    	First, suppose $a'_i\in\textrm{BR}_i(a_{-i},U)$, then we have both $a'_i\in\textrm{BR}_i(a_{-i},\bar{U})$ and $|\textrm{BR}_i(a_{-i},\bar{U})|\leq |\textrm{BR}_i(a_{-i},U)|$.
    	The first result follows intuitively as only the utility of $a'_i$ increased from $U_i$ to $\bar{U}_i$ and thus $a'_i$ must remain in the best response set.
    	The second follows as if $l>0$ then $a'_i$ becomes the unique best response and the best response set remains unchanged if $l=0$.
    	With both of these results together, we conclude $P_i(a,a',\bar{U})\geq P_i(a,a',U)$ from the definition of \eqref{eq:best response dynamics}.
    	The other case, $a'_i\notin\textrm{BR}_i(a_{-i},U)$, is trivial as $P_i(a,a',U)=0$ is a lower bound of the nonnegative function $P(\cdot,\bar{U})$.
    	
    	Locality for both log-linear learning and best response dynamics can be seen directly from their definitions, as they only dependent on selected action $a'_i$ and utilities of the form $U(\cdot,a_{-i})$ for prior action profile $a$.
        Thus, log-linear learning and best response learning are individual learning rules, monotone with respect to utility, and local.
        \end{proof}

\section{Main Contribution}

    To characterize history-dependent game $g$ using a reference game $\hat{g}$ as per the definition of aligned history-dependent games, we develop a monotone coupling between $P_\pi$ and $\hat{P}_\pi$.
    Because developing the coupling is technically involved, we defer it to Section~\ref{sec:proofs};
    in this section we present the game theoretic significance of coupling.
    {We begin by giving a broad result which follows immediately from the existence of the coupling to relate increasing metrics in the reference game to their counterparts in the history-dependent one.}
    \begin{theorem} \label{thm:z functions}
        Let $g$ be an aligned history-dependent game, $P$ be a local and monotone asynchronous learning rule, and $Z:\mathcal{A}_T\rightarrow\mathbb{Z}^+$ be an increasing function with respect to ordering $\geq_{\mathcal{A}_T}$. Then, we have
        \begin{equation}
            \mathbb{E}_{P_\pi}(Z)\geq\mathbb{E}_{\hat{P}_\pi}(Z)
        \end{equation} for any $\pi \in \Delta(A),~T\in\mathbb{N}.$
    \end{theorem}
The proof of Theorem \ref{thm:z functions} will be given in Section~\ref{sec:proof of z functions}.

    \begin{figure}[!htbp]
        \centering
        \includegraphics[scale=0.6]{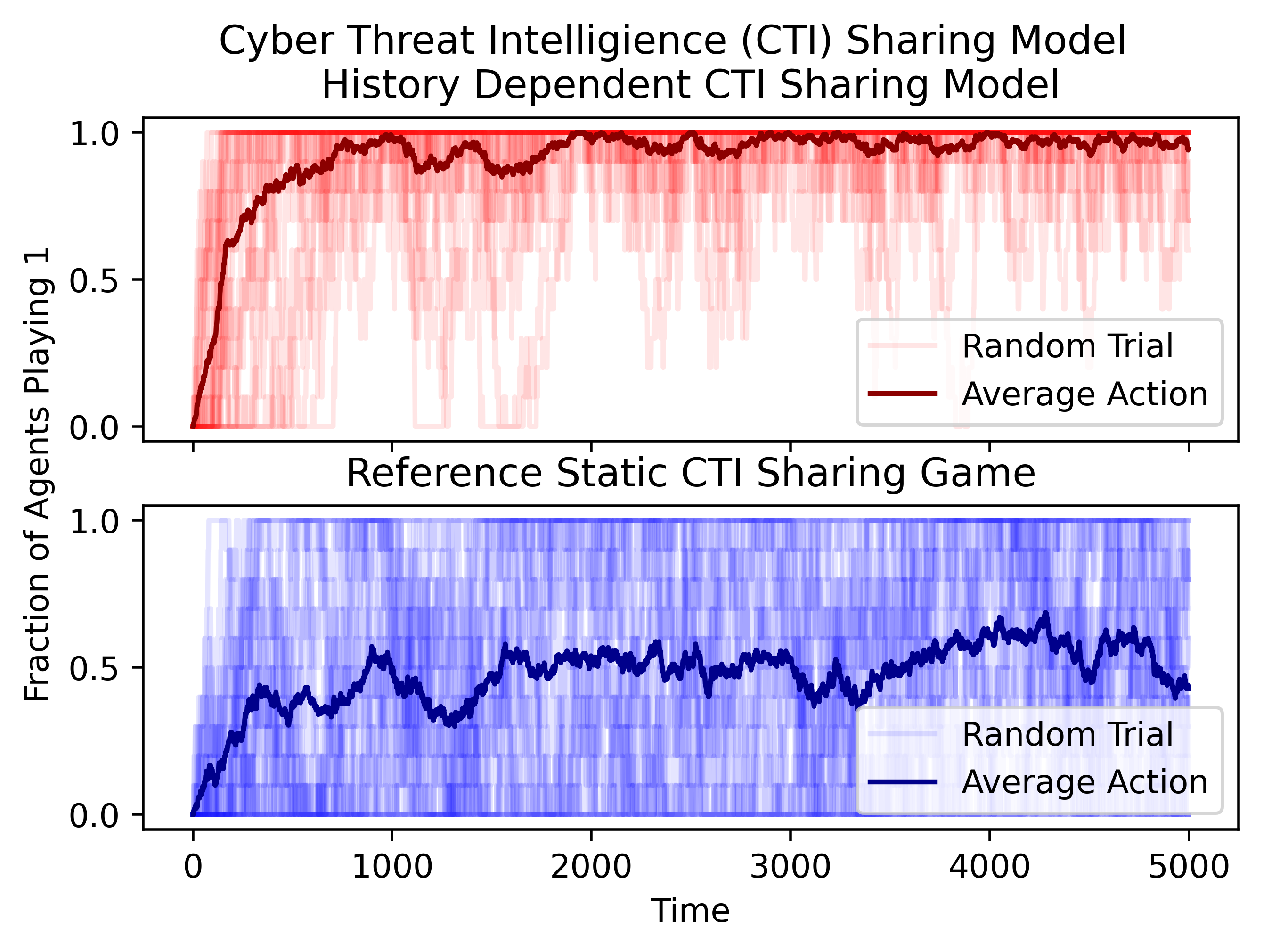}
        \caption{Comparison of Cyber Threat Intelligence Sharing Game between a history-dependent game and the reference static game with log-linear learning. This experiment was performed on a ring graph of 10 nodes with 25 random trials and a temperature $\tau=0.1$. The full details of this experiment can be found in Section~\ref{sec:CTI}.}
        \label{fig:CTI}
    \end{figure}
    \begin{figure}
        \centering
        \includegraphics[scale=0.6]{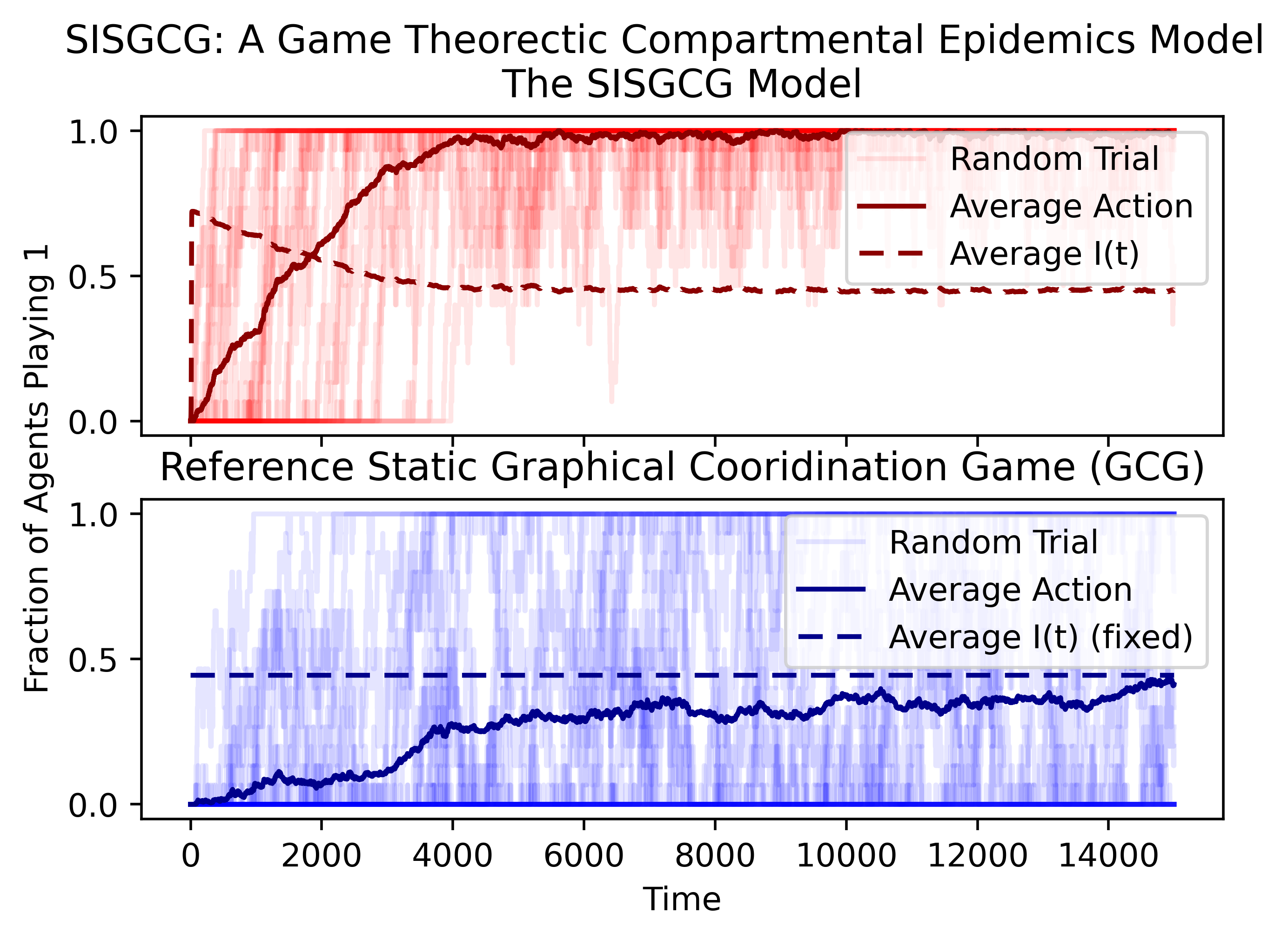}
        \caption{Comparison between a history-dependent SISGCG model and the reference static GCG model with lg-linear learning. This experiment was performed on a ring graph of 15 nodes with 40 random trials and a learning temperature $\tau=0.3$. 
        The dotted lines show the average disease $I(t)$ across all trials at time step $t$. The full details of this experiment can be found in Section~\ref{sec:epidemics}.}
        \label{fig:SISGCG}
    \end{figure}

The above $Z$ functions can represent a wide range of metrics of interest.
    For instance, $Z$ could be the total time spent in the $\vec{1}$ state or the total time during which some proportion of agents were playing 1.
    Thus, Theorem~\ref{thm:z functions} gives a lower bound on the expected value of a broad class of increasing metrics on the history-dependent game of interest.
Figures~\ref{fig:CTI}-\ref{fig:SISGCG} show this intuition by using two specific instances of aligned history-dependent games, whose details are deferred to Section~\ref{sec:CTI} and Section~\ref{sec:epidemics}, respectively.
In both plots, the faint red curves show trials of an aligned history-dependent game and the faint blue curves show the static reference game, both under log-linear learning.
It can be seen in both models that trials of the history-dependent game frequently feature more agents playing the $1$ action than in the static game, 
{confirming}
Theorem~\ref{thm:z functions} over a broad class of monotone metrics $Z$. 
    
To develop further applications of Theorem~\ref{thm:z functions}, we consider a specific $Z$ function.
    In doing so we obtain a property known as {\em stochastic dominance} as shown in the following lemma.
	\begin{lemma} \label{thm:stochastic dominance}
		If $g\in \mathcal{G}^{\mathcal{A}}$ is an aligned history-dependent game with corresponding static game $\hat{g}$, $P$ is a local and monotone asynchronous learning rule, and $\mathcal{I}\subseteq \mathcal{A}_T$ is an upper set (i.e., $\alpha\in\mathcal{I}\implies \alpha'\in\mathcal{I}$ if $\alpha'\geq_{\mathcal{A}_T}\alpha$), then $P_\pi(\mathcal{I})\geq \hat{P}_\pi(\mathcal{I})$.
    \end{lemma}
    \begin{proof}
        Let $g\in\mathcal{G}^{\mathcal{A}}$ be an aligned history-dependent game, $P$ be a local and monotone asynchronous learning rule, and $\mathcal{I}\subset\mathcal{A}_T$ be an upper set.
        Define $\mathds{1}_\mathcal{I}(\alpha):=\mathds{1}(\alpha\in\mathcal{I})$ as an indicator function.
    	By definition of expectation we have 
    	\begin{equation}
    	    P_\pi(\mathcal{I})=\mathbb{E}_{P_\pi}(\mathds{1}_\mathcal{I})\mbox{, }\hat{P}_\pi(\mathcal{I})=\mathbb{E}_{\hat{P}_\pi}(\mathds{1}_\mathcal{I}).
    	\end{equation}
    	Because $\mathds{1}_\mathcal{I}$ is increasing we apply Theorem~\ref{thm:z functions} to the above to obtain
    	\begin{equation}
    	    P_\pi(\mathcal{I})\geq \hat{P}_\pi(\mathcal{I}).
    	\end{equation}
    \end{proof}

    This result can then be interpreted as a lower bound on the probability of any upper set $\mathcal{I}$ or collection of histories occurring in the aligned history-dependent game of interest.
    Although it can be difficult to interpret a bound of the probability of an upper set, a specific choice of $\mathcal{I}$ yields a powerful inequality derived in the following lemma. 
    \begin{lemma}\label{thm:proof of stoch dom+}
        If $g\in\mathcal{G}^\mathcal{A}$ is an aligned history-dependent game and $P$ is a local and monotone asynchronous learning rule, then  $\Pr(s(T;P_\pi)=\vec{1})\geq \Pr(s(T;\hat{P}_\pi)=\vec{1})$ for any  $\pi\in\Delta(A),T\in\mathbb{N}$.
    \end{lemma}
    \begin{proof}
Let $g\in\mathcal{G}^\mathcal{A}$ be an aligned history-dependent game and let $P$ be a local and monotone asynchronous learning rule.
Define $\mathcal{I}$ as the upper set such that $((\vec{0})^{T-1}_{t=1},\vec{1})\in \mathcal{I}$.
This induces $\mathcal{I}$ such that it includes every path such that at time $T$ the $\vec{1}$ state is played.
		This yields the following interpretation: 
		\begin{equation}\label{eqn:probability of 1}
		P_\pi(\mathcal{I})=\mbox{Pr}(s(T;P_\pi)=\vec{1})
		\end{equation}
		is the probability that at time $T$ game $g$ is in the $\vec{1}$ action profile given initial distribution $\pi\in\Delta(A)$, while noting 
		that a parallel interpretation to \eqref{eqn:probability of 1} holds for $\hat{P}_\pi$. We apply these to Lemma~\ref{thm:stochastic dominance} to obtain
		\begin{equation}\label{eq:asynch cournot bound}
		    \mbox{Pr}(s(T;P_\pi)=\vec{1})\geq\mbox{Pr}(s(T;\hat{P_\pi})=\vec{1}).
		\end{equation}
    \end{proof}

This result gives a lower bound on the probability that the history-dependent game is in $\vec{1}$ at any time $T$.
    The significance of this inequality is that game $\hat{g}$ can be analyzed using traditional game theoretic techniques and the resulting characterization can be applied using this inequality.
    We give two results of this flavor, first on static games with best response dynamics and then on potential games with log-linear learning.
    We give the first of these results in the following theorem.
    \begin{theorem}\label{thm:asynchronus cournot results}
        Let $g\in \mathcal{G}^{\mathcal{A}}$ be an aligned history-dependent game and $P$ be asynchronous best response as specified by \eqref{eq:asynchronus learning rule}
        with individual learning rules given by \eqref{eq:best response dynamics}.
        If the associated game $\hat{g}$ is a reference static game with $\vec{1}$ being a strict Nash
        equilibrium,  then
        \begin{equation}
            \mbox{Pr}(s(t;P_\pi)=\vec{1})\geq\mbox{Pr}(s(t;\hat{P_\pi})=\vec{1})\mbox{ for all }t,\pi.
        \end{equation}
        Further, if $\pi$ has full support over $A$, then 
        \begin{equation}\label{eq:asynch cournt bounded above 0}
            \mbox{Pr}(s(t;P_\pi)=\vec{1})>0 \mbox{ for all }t.
        \end{equation}
    \end{theorem}
    \begin{proof}
    Let $g\in \mathcal{G}^{\mathcal{A}}$ be an aligned history-dependent game, $P$ be asynchronous best response as specified by \eqref{eq:asynchronus learning rule} with individual learning rules given by \eqref{eq:best response dynamics}, and the associated game $\hat{g}$ be a static game with $\vec{1}$ being a strict Nash equilibrium.
    Eq.\eqref{eq:asynch cournot bound} holds directly from Lemma~\ref{thm:proof of stoch dom+} 
    and the fact that Lemma~\ref{thm:individual, monotone, local} gives that $P$ is local and monotone.
    Further, supposing $\pi$ has full support over $A$, \eqref{eq:asynch cournt bounded above 0} follows directly from this, \eqref{eq:cournot limit}, and \eqref{eq:asynch cournot bound}.
    Specifically, \eqref{eq:cournot limit} can be applied because $\vec{1}$ is a strict Nash equilibrium.
    In combination with that $\pi$ have full support over $A$, we have
    \begin{equation}
        \Pr(s(t,\hat{P}_\pi)\geq\Pr(s(0,\hat{P}_\pi)=\vec{1})=\pi(\vec{1})>0
    \end{equation}
    as desired.
    \end{proof}
    
    We further develop results on potential games with log-linear learning, by applying the theory of potential games to a history-dependent game.
    \begin{theorem}\label{thm:SS}
        Let $g\in \mathcal{G}^{\mathcal{A}}$ be an aligned history-dependent game and $P$ be log-linear learning as specified by \eqref{eq:asynchronus learning rule} with individual learning rules given by \eqref{eq:log-linear learning}.
        If the associated game $\hat{g}$ is an exact potential game with $\vec{1}$ being the sole maximizer of the potential function, then $\vec{1}$ is uniquely stochastically stable in $g$ under $P$.
    \end{theorem}
    \begin{proof}
        Let $g\in \mathcal{G}^{\mathcal{A}}$ be an aligned history-dependent game and $P$ be a local and monotone asynchronous learning rule. Let the associated game $\hat{g}$ be an exact potential game with $\vec{1}$ being the sole maximizer of the potential function $\phi$.
        It is well-known~\cite{alos2010logit} that in an exact potential game $a\in A$ is a stochastically stable state with log-linear learning if
    	\begin{equation}
    	    a\in\argmax_{a'\in A}\phi(a').\vs\vs
    	\end{equation}
        Because $\vec{1}$ is the sole maximizer of $\phi$, it is strictly stochastically stable.
	    Because $P$ is local and monotone, we apply Lemma~\ref{thm:proof of stoch dom+} to the definition of strict stochastic stability given by~\eqref{eqn:SS to g hat}.
		For any $\epsilon>0$, there exists $\mathcal{T}>0,T<\infty$ such that\vs
		\begin{equation}\vs
			\mbox{Pr}(s(t;P_\pi)=\vec{1})\geq\mbox{Pr}(s(t;\hat{P}_\pi)=\vec{1})>1-\epsilon\; \vs
		\end{equation}
		for all $t>T,\tau<\mathcal{T}$, yielding stochastic stability of $\vec{1}$ in game $g$.
    \end{proof}

\section{Sample Applications of the Aligned History-Dependent Framework}

\subsection{Cyber Threat Intelligence Sharing}
\label{sec:CTI}

Cyber Threat Intelligence (CTI) refers to any information pertaining to, for example, why and how a firm (or enterprise) has been attacked, including 
the techniques or tactics that are used by attackers.
In principle, sharing CTI is mutually beneficial to firms as it allows firms to strengthen their cyber defense postures,
In practice, firms frequently opt not to share CTI due to a myriad of reasons including competitor relationships with other firms and possible leakage of their internal cybersecurity posture.

The CTI sharing problem has been studied using game theory between two firms in \cite{hausken2007,gao2014,gao2016differential,DBLP:conf/icc/ToshSKKM15}, between many firms using a centralized repository in \cite{solak2020optimal}, and in a decentralized (i.e., networked) setting in \cite{collins2021paying}.
However, one assumption common to all these studies is that the value of CTI is constant over time.
This oversimplifies the problem because the value of CTI  can vary wildly.
For example, detecting a zero-day exploit is a much more valuable CTI than the detection of a known attack.
In this example, we propose a networked model where the value of CTI may vary over time, which can reflect the history of sharing. We examine how the aligned history-dependent framework can be used to analyze such a model.

Let a group of firms be modeled by the set $N=\{1,2,\ldots,n\}$.
If two firms mutually value each other's CTI (i.e. because they use similar technologies), then we consider them connected in the networked model.
We denote this with edge set $E$ such that firms $i,j\in N$ are connected if $(i,j)\in E$.
Then, $G=(N,E)$ represents an undirected graph which we leverage to model how firms decide to share CTI.
We use action space $A$ to represent a firm's decision to share CTI; in particular, a firm chooses action $a_i\in A_i=\{0,1\}$ where $a_i=1$ means firm $i$ wishes to share CTI with its neighbors and $a_i=0$ indicates firm $i$ does not share CTI.

We make the following assumptions on the mechanics of sharing CTI.
First, CTI is only shared if two neighbor firms $i,j$ mutually wish to share CTI, meaning $a_i=a_j=1$.
This is because CTI is potentially sensitive and thus firms would guard their CTI closely and require reciprocity to share it with other firms.
Second, CTI sharing incurs firms a fixed infrastructure cost $c_i\geq0$.
This cost represents all expenses and labor associated with creating, packaging and sharing CTI and is incurred whenever firm $i$ selects $a_i=1$.
Third, the value of firm $i$'s CTI, given by $v_i:\mathcal{A}\rightarrow \mathbb{R}_+$, is a function of the history of agent behavior.
That is, when firms $i,j$ share CTI given history $\alpha\in\mathcal{A}$, firm $i$ receives CTI from $j$ with value $v_j(\alpha)$ and firm $j$ receives CTI from $i$ with value $v_i(\alpha)$.

{The cybersecurity meaning of
$v_i(\alpha)$ can be interpreted in as follows.
First, 
it is possible that the value of CTI is a function of time, say $v_i(\alpha)=v'_i(T)$, where $v'_i$ is a function of time and $T$ is the period of time corresponding to history $\alpha$. For example, it is possible that
firm $i$ occasionally discovers and detects important pieces of CTI, which correspond to a high value $v'_i(T)$, but most often discovers and shares less valuable pieces of CTI, meaning a low $v'_i(T)$ for most $T$. Second, 
the value of CTI can vary based on the history of play as shown in the following two scenarios.
(i) As firms gain experience sharing CTI, their ability to share and leverage other firms' CTI matures.
This can be modeled by $v_i(\alpha)$ which increases over time or as the history $\alpha$ evolves.
(ii) Attackers often attempt to disrupt defenses, such as CTI sharing in this case.
As a consequence, attackers may wage new attacks that are harder to detect and share after the employed attacks become easy to detect.
}

The preceding discussion leads to the following utility function:
\begin{equation}\label{eq:CTI utility}
    U^\alpha_i(a)=a_i\left(-c_i+\sum_{j\in\mathcal{N}_i}a_j v_j(\alpha)\right),
\end{equation}
where $\mathcal{N}_i=\{j\in N\mid (i,j)\in E\}$ is the neighbor set of firm $i$ in the CTI sharing graph $G=(N,E)$.
{Since the value of CTI can vary in possibly infinitely many ways, to facilitate analysis we consider a single family of them, by 
imposing a mild restriction 
on the lower bound of $v_i$:}
\begin{equation}\label{eq:definition of underbar v}
    \underbar{v}:=\min_{i\in N}\inf_{\alpha\in\mathcal{A}}v_i(\alpha).
\end{equation}
{This is a lower bound of all firms 
in all histories $\alpha$. Under this premise,
we show that CTI sharing game $g$ is an aligned history-dependent game as follows.}
\begin{proposition}\label{thm:CTI sharing is AHDG}
    If $g\in\mathcal{G}^\mathcal{A}$ be a CTI sharing game with utilities specified by \eqref{eq:CTI utility}, then $g$ is an aligned history-dependent game.
\end{proposition}

\begin{proof}
Let $g\in\mathcal{G}^\mathcal{A}$ be a CTI sharing game.
We begin by giving the static game defined by
\begin{equation}\label{eq:static CTI utility}
    \hat{U}_i(a)=a_i\left(-c_i+\sum_{j\in\mathcal{N}_i}a_j \underbar{v}\right).
\end{equation}
We now show
\begin{enumerate}
	\item $U_i(1,\alpha^T_{-i};\alpha)\geq\hat{U}_i(1,a_{-i})$, and \label{eq:AHDG:1}
	\item $\hat{U}_i(0,a_{-i})\geq U_i(0,\alpha^T_{-i};\alpha)$ \label{eq:AHDG:0}
\end{enumerate}
for any $\alpha,\in\mathcal{A}$, $a,a'\in A$, $T\in\mathbb{N}$ where $\alpha^T_{-i}\geq_{A_{-i}} a_{-i}$ and $a,a'$ differ only by a unilateral deviation.

Note that condition~\ref{eq:AHDG:1} follows from \eqref{eq:definition of underbar v} and the fact that if firm $i$ plays $1$, then firm $i$'s neighbors playing $1$ 
can only 
increase firm $i$'s utility.
Condition~\ref{eq:AHDG:0} is trivial as $\hat{U}_i(0,a_{-i})= U_i(0,\alpha^T_{-i};\alpha)=0$.
\end{proof}

Now that the aligned history-dependent game framework has been established on the CTI game, we show how it can be used to analyze best response dynamics and log-linear learning in a history-dependent game.
Before developing these results, we show that the static game $\hat{g}$ is a potential game.
\begin{proposition}\label{thm:CTI potential game}
    The static game $\hat{g}$ is an exact potential game.
\end{proposition}
\begin{proof}
We first show that static game $\hat{g}=(N,A,\hat{U})$ with $\hat{U}$ defined by \eqref{eq:static CTI utility} is an exact potential game by giving its potential function:
\begin{equation}\label{eq:CTI sharing potetial func}
    \phi(a)=\sum_{i\in N}\left((1-a_i)c_i+\frac{a_i}{2}\sum_{j\in\mathcal{N}_i}a_j\underbar{v}\right).
\end{equation}
The potential function can be verified by checking \eqref{eq:potential function definition}.
Letting $\mathcal{N}_i(1)$ denote the neighbors of agent $i$ with play action 1. We compute the change in potential and utility if agent $i$ switches from $0$ to $1$ given arbitrary action profile $a_{-1}$:
\begin{equation}
\begin{aligned}
        \phi(1,a_{-i})-\phi(0,a_{-i})&=\underbar{v}|\mathcal{N}_i(1)|-c_i\\
        &=U_i(1,a_{-i})-U_i(0,a_{-i})
\end{aligned}
\end{equation}
and the situation where agent $i$ switches from $1$ to $0$ is simply the above equality multiplied by $-1$.
\end{proof}

In the case of the best response dynamics, it is known that the process does not leave Nash equilibrium.
By establishing these properties in the static game, we establish an analogous result in the history-dependent game by applying the framework.
Particularly, we establish that the probability that $\vec{1}$ is selected by the best response dynamics in the history-dependent game is lower bounded by the static game, as given below.
\begin{proposition}
    Let $g\in\mathcal{G}^{\mathcal{A}}$ be a CTI sharing game, $P$ be the best response dynamics defined by \eqref{eq:best response dynamics}, distribution $\pi$ have full support over $A$, and $\hat{g}$ be its associated static game defined by $\underbar{v}$.
    If
    \begin{equation}
        \underbar{v}>\frac{c_i}{|\mathcal{N}_i|} \textrm{ for all }i\in N.
    \end{equation}
    then
    \begin{enumerate}
        \item $\vec{1}$ is a strict Nash equilibrium in $\hat{g}$,
        \item $\Pr(s(t,P_\pi)=\vec{1})\geq\Pr(s(t,\hat{P}_\pi)=\vec{1})$, and
        \item $\Pr(s(t,P_\pi)=\vec{1})>0$ for all $t$.
    \end{enumerate}
\end{proposition}
\begin{proof}
    It is easy to see that $\vec{1}$ is a Nash equilibrium in $\hat{g}$, as any deviation leads to a utility loss
    \begin{equation}
        u_i(\vec{1})-u_i(0,\vec{1}_{-i})=|\mathcal{N}_i|\underbar{v}-c_i> 0.
    \end{equation}
    The second and third results hold from Theorem~\ref{thm:asynchronus cournot results} as $\vec{1}$ is a strict Nash equilibrium in $\hat{g}$, $g$ is an aligned history-dependent game by Proposition~\ref{thm:CTI sharing is AHDG}, and $\pi$ has full support over $A$.
    %
\end{proof}


Next, we examine the CTI sharing game under log-linear learning.
As the long run behavior of potential games is well understood under log-linear learning, we can extend these results from the static game to the history-dependent one.
By establishing a sufficient condition on $\underbar{v}$ such that $\vec{1}$ is a unique maximizer of $\phi$, we can apply Theorem~\ref{thm:SS} to obtain stochastic stability in the history-dependent CTI game.

\begin{proposition}\label{thm:CTI and SS}
Let $g\in\mathcal{G}^{\mathcal{A}}$ be a CTI sharing game, $P$ be the log-linear learning dynamics defined by \eqref{eq:log-linear learning}, and $\hat{g}$ be its associated static game defined by $\underbar{v}$.
    If
    \begin{equation}\label{eq:underbarv is potential func maximizer}
        \underbar{v}>\frac{2c_i}{|\mathcal{N}_i|} \textrm{ for all }i\in N,
    \end{equation}
    then we have
    \begin{enumerate}
        \item $\{\vec{1}\}=\argmax_a \phi(a)$, and
        \item $\vec{1}$ is strictly stochastically stable in $g$ with log-linear learning.
    \end{enumerate}
\end{proposition}
\begin{proof}
    Let $g\in\mathcal{G}^{\mathcal{A}}$ be a CTI sharing game, $P$ be the log-linear learning dynamics defined by \eqref{eq:log-linear learning}, and $\hat{g}$ be its associated static game defined by $\underbar{v}$, and $\phi$ be $\hat{g}$'s potential function defined in \eqref{eq:CTI sharing potetial func}.
    To show that $\vec{1}$ is the sole maximizer of $\phi$, we show $\phi(\vec{1})>\phi(a)$ where $a$ is any action profile $a\neq\vec{1}$ as follows. 
    \begin{equation}
    \begin{aligned}
        \phi(\vec{1})-\phi(a)&=\frac{1}{2}\sum_{(i,j)\in E}\underbar{v}-\sum_{i\in N} (1-a_i)c_i-\frac{1}{2}\sum_{(i,j)\in E}a_i a_j \underbar{v}\\
        &=\frac{1}{2}\sum_{(i,j)\in E}(1-a_i a_j)\underbar{v}-\sum_{i\in N} (1-a_i)c_i\\
        &\geq \frac{1}{2}\sum_{(i,j)\in E}(1-a_i)\underbar{v}-\sum_{i\in N} (1-a_i)c_i\\
        &=\frac{1}{2}\sum_{i\in N}(1-a_i)|\mathcal{N}_i|\underbar{v}-\sum_{i\in N} (1-a_i)c_i\\
        &>\sum_{i\in N}(1-a_i)(\frac{2|\mathcal{N}_i|c_i}{2|\mathcal{N}_i|}-c_i)=0,
    \end{aligned}
    \end{equation}
    where the first inequality holds because $a_i\geq a_i a_j$ in domain $\{0,1\}$ and the second inequality is an application of 
    \eqref{eq:underbarv is potential func maximizer} 
    because there exists $i$ such that $a_i=0$ by the definition of $a$.
\end{proof}

A numerical comparison of the behavior of the CTI sharing game and the corresponding static game with log-linear learning is shown in Figure~\ref{fig:CTI}.
Both models consider a ring graph of 10 firms, and costs $c_i=0.4$ for all $i\in N$.
The history-dependent game model assumes a random model of CTI, where the 
value of CTI is a random value that is bounded from below, as given by
\begin{equation}\label{eq:CTI example}
v(\alpha)=\max_{i\in N}\left(\frac{2c_i}{|\mathcal{N}_i|}\right)+X+\epsilon=0.4+X+\epsilon,
\end{equation}
where $X$ is a uniform random variable in the interval $[0,0.1]$ and $\epsilon$ is a positive constant. 
Thus, the value of CTI is random but sits in interval $[0.4+\epsilon,0.5+\epsilon]$.
This induces $\underbar{v}=0.4+\epsilon$, defining the static game.
Since $\underbar{v}$ satisfies \eqref{eq:underbarv is potential func maximizer}, Proposition~\ref{thm:CTI and SS} holds, meaning that $\vec{1}$ is stochastically stable in both the reference game and the history-dependent game.
This is numerically confirmed 
as Figure~\ref{fig:CTI} shows that the CTI sharing game rapidly finds the $\vec{1}$ state and stays there frequently.
Although the static game also has $\vec{1}$ as the sole stochastically stable state, $\vec{1}$ only maximizes the potential function by an $\epsilon$ margin ($\epsilon=0.001$ in this case) and thus 
the static game does not clearly frequently spend time in $\vec{1}$ state.
As per the definition of {\em stochastic stability}, if we take $\tau\rightarrow 0$ and $t\rightarrow\infty$ numerically,
we can expect to see the static game tends toward $\vec{1}$ almost surely.
Note that the key {requirement for 
applying} the aligned history-dependent framework is that $v(\alpha)$ is bounded from below, regardless of the exact form of \eqref{eq:CTI example}, random or deterministic alike.

\subsection{Epidemics}
\label{sec:epidemics}

To show Theorem~\ref{thm:SS}'s usefulness in analyzing the stochastic stability of history-dependent games, we consider another example, which is a
simple model of epidemics.
One challenge of epidemic modeling is to account for the interplay between epidemic severity and (in this example) the voluntary adoption of preventative social conventions. For instance, in the absence of an epidemic, people prefer not to wearing masks; however, in a widespread epidemic people may prefer to take preventative measures. To model this phenomenon, we intertwine the SIS compartmental epidemic model and the graphical coordination game (GCG) which models the spread and adoption of the relevant preventative social conventions; we term this model SISGCG.

The fraction of individuals in the society susceptible to infection is described by the nonlinear differential equation
	\begin{equation}\label{eqn:SISGCG}
		\dot{S}=(1-S)(\gamma-\beta(t) S),
	\end{equation}
where $\gamma>0$ is the curing rate and $\beta(t)>0$ is the rate of infection which depends on agent actions.
The action $1$ represents a ``safe convention'' action in which an agent acts to reduce contagion (e.g., wearing a mask); the action $0$ represents conventions ignoring the pandemic.
These actions are associated with infection coefficients $0<\beta_1<\beta_0$, respectively.
	Accordingly, $\beta(t)$ is simply the average infection rate of all individuals, 
	\begin{equation}\label{eqn:betat}
		\beta(t)=\frac{1}{|N|}\sum_{i\in N}a^t_i\beta_1+(1-a^t_i)\beta_0,
	\end{equation}
	where $a_i^t$ is the action selected by agent $i$ at time $t\in\mathbb{N}$.
    Actions are selected by agents in $N$ dynamically on undirected graph $G=(N,E)$ according to log-linear learning rule~\eqref{eq:individual one step with history}. 
	The utility of agent $i$ at time $t$ is given by
	\begin{equation}\label{eqn:uSISGCG}
		U_i^\alpha(a_i^t,a_{-i}^t)=a_i|\mathcal{N}_i(1)|(\lambda+I(t))+(1-a_i)|\mathcal{N}_i(0)|,
	\end{equation}
	where $\mathcal{N}_i(x)=\{j\in N\mid (i,j)\in E, a_j=x\}$ is the set of $i$'s neighbors which play $x\in \{0,1\}=A_i$, $I(t) := 1-S(t)$ is the fraction of infected individuals, and $\lambda\in (0,1]$ represents the agent's willingness to practice safe conventions in the absence of epidemic.
	
	The SISGCG model can be analyzed using the aligned history-dependent game framework.
	Specifically, a reference static game can be devised using utility function \eqref{eqn:uSISGCG} by setting $I(t)=1-\gamma/\beta_1$, which is a lower bound of $I(t)$ after a sufficient time.
	The details of this lower bound are given in the following Proposition:
	\begin{proposition}\label{thm:positive invariant}
	    If $S(0)\in[0,1)$ and $S(t)$ is a solution to~\eqref{eqn:SISGCG} with $\beta(t)$ being given by~\eqref{eqn:betat}, then there exists a $\bar{t}$ such that $S(t)\leq \gamma/\beta_1$ for all $t\geq \bar{t}$ almost surely.
	\end{proposition}
	\begin{proof}
	    We write $S_1^*:= \gamma/\beta_1.$
	    Note that if $S(t)\geq S_1^*$, then because $\beta(t)\geq\beta_1$, we have that $\dot{S}\leq0$ by~\eqref{eqn:SISGCG} and that this inequality is strict when $S(t)> S_1^*$.
	    Thus, the set $[0,S_1^*]$ is positively invariant for the hybrid nonlinear dynamics given in~\eqref{eqn:SISGCG}.
	    
	    To see that $S(t)$ eventually enters $[0,S_1^*]$ almost surely, consider the event that $S(t)>S_1^*$ for all $t$.
	    Note that $S_1^*$ is asymptotically stable when $\beta(t)\equiv \beta_1$.
	    For any action profile $a\neq\vec{1}$ such that its associated $\beta>\beta_1$,
	    the event that $\beta(t)\equiv \beta_1$ for all $t$ is the same event as $S(0)>S_1^*$ and $S(t)>S_1^*$ for all $t$.
	    However, it can be seen that log-linear learning defines a stochastic process which visits every action profile in $A$ infinitely often.
	    That is, the probability that $\beta(t)\equiv \beta_1$ is $0$, and thus there must exist a $\bar{t}$ such that $S(t)\leq S_1^*$ for all $t\geq\bar{t}$ almost surely.
	\end{proof}

    It can be seen from \eqref{eqn:uSISGCG} that SISGCG can be represented by a history-dependent game, as the utility function depends on the history of play, so
    Theorem~\ref{thm:SS} allows us to reference SISGCG to a related exact potential game and deduce conditions guaranteeing that $\vec{1}$ is strictly stochastically stable under log-linear learning.
	\begin{proposition}\label{thm:sisgcg is ss}
		Let $g^S$ be an instance of SISGCG, and $P$ be log-linear learning as defined by \eqref{eq:log-linear learning} and \eqref{eq:asynchronus learning rule}.
		If $\beta_1/\gamma>1$, $\lambda+\gamma/\beta_1>1$, and $I(0)>0$, then $\vec{1}$ is stochastically stable in $g$ with $P$.
	\end{proposition}
	\begin{proof}
        Denote the SISGCG model by $g^S$, which is played on graph $G=(N,E)$ with $q+\gamma/\beta_1>1$ and $I(0)>0$. Consider $g^S$ as played after time $\bar{t}$ as given by Proposition~\ref{thm:positive invariant}.
		Game $g^S$ is a history-dependent game since~\eqref{eqn:uSISGCG} depends on $I(t)$, which is a function of the history $\alpha$.
		Thus, we have $g^S=(N,A,U)\in\mathcal{G}^{\mathcal{A}}$ where $U=\{U^\alpha\mid \alpha\in\mathcal{A}\}$.
		
		Now we let $\hat{g}^S=(N,A,\hat{U}^S)$ be a graphical coordination game played on graph $G$, where the utility function $\hat{U}^S$ is given by~\eqref{eqn:uSISGCG} with $I(t)= 1-\gamma/\beta_1$.
		Standard results give that $\hat{g}^S$ is an exact potential game and that $\vec{1}$ is its sole potential function maximizer~\cite{young1993evolution}.
%
	    %
    	
    	We now use $\hat{g}^S$ to show that $g^S$ is an aligned history-dependent game.
    	We verify $U_i^\alpha(1,\alpha^T_{-i})\geq \hat{U}^S_i(1,a_{-1})$ holds for $\alpha^T_{-i}\geq_{A_{-i}}a_{-i}$ and $t>\bar{t}$.
    	It can be rewritten for $t>\bar{t}$ as
    	\begin{equation}
    	    \sum_{j\in\mathcal{N}_i(1;\alpha^T_{-i})}\lambda+I(t)\geq\sum_{j\in\mathcal{N}_i(1;a_{-i})}\lambda+\gamma/\beta_1,
    	\end{equation}
    	where $\mathcal{N}_i(1;a_{-i})$ denotes the neighbors of $i$ which play $1$ given profile $a$.
    	This expression holds because $\alpha^T_{-i}\geq_{A_{-i}} a_{-i}\Rightarrow |N_i(1;\alpha^T_{-i})|\geq |N_i(1;a_{-i})|$ and by Proposition~\ref{thm:positive invariant}. 
    	An argument with the same structure holds for $U_i^\alpha(0,\alpha^T_{-i})\leq \hat{U}^S_i(0,a_{-1})$.
    	Thus, $g^S$ is an aligned history-dependent game, and Theorem~\ref{thm:SS} gives that $\vec{1}$ is strictly stochastically stable. 
	\end{proof} 
    
Figure~\ref{fig:SISGCG} plots a numerical example of the SISGCG and its associated static game,
where 
$\gamma=0.25$,
$\beta_0=0.9$, $\beta_1=0.45$, and $\lambda=\gamma/\beta_1+\epsilon$.
It can be seen that the average $I(t)$ across SISGCG rapidly shoots above $1-\gamma/\beta_1$, which is the
fixed disease $I(t)$ for the GCG game.
This confirms the intuition of Proposition~\ref{thm:positive invariant} that $1-\gamma/\beta_1$ can be treated as an effective lower bound of $I(t)$ which we use to define the static game. 
Figure~\ref{fig:SISGCG} shows that the SISGCG model rapidly finds the all $\vec{1}$ state faster and more consistently than the reference game,
which supports
Proposition~\ref{thm:sisgcg is ss} experimentally.
Similar to the above CTI sharing example, $\vec{1}$ is only stochastically stable by an $\epsilon$ margin, explaining the phenomenon it does not appear to find the $\vec{1}$ state frequently.
Critically, the history-dependent game framework can be applied to characterize agent behaviors in a stochastically complex epidemic model.
	
\section{Proving the Monotone Coupling}
\label{sec:proofs}
\subsection{A Primer on Monotone Couplings}
    We begin with the definition of {\em monotone coupling}, the core analytical device for this paper.
	\begin{definition}\label{def:monotonic coupling}
			Let $X$ be a countable set with partial ordering $\leq_X$ and $p_1,p_2$ be probability measures on 
			$(X,\mathcal{F})$, where $\mathcal{F}$ is a $\sigma$-algebra of $X$. 
			Then, a \emph{monotone coupling} of $p_1,p_2$ is a probability measure $p$ on $(X^2,\mathcal{F}^2)$ satisfying the following for all $x,y\in X$ 
		\begin{equation}\label{eqn:monotonic coupling}
				\sum_{x\leq_{X}y'} p(x,y')=p_2(y') \mbox{ and }\sum_{y\geq_Xx'}p(x',y)=p_1(x').
		\end{equation}
 
	\end{definition}
	
	A monotone coupling is a useful tool for analyzing the component probability measures $p_1$ and $p_2$.
	In particular, the following property holds for monotone couplings.
    \begin{proposition}[Paarporn et al.,~\cite{paarporn2017}]\label{thm:increasing functions in expectation}
        Let $p_1,p_2$ be probability measures on $(X,\mathcal{F})$.
        If $p$ is a monotone coupling of $p_1,p_2$, then for any increasing random variable $Z:X\rightarrow \mathbb{Z}_+$ we have
        \begin{equation}\label{eq:coupling equality}
            \mathbb{E}_{p_1}(Z)-\mathbb{E}_{p_2}(Z)=\sum_{\eta=0}^\infty p(Z^c_\eta,Z_\eta),
        \end{equation}
        where $Z_\eta=\{a\mid Z(a)>\eta\}$ and $Z_\eta^c:=X\setminus Z$ is the complement set of $Z\subseteq X$.
    \end{proposition}

	
\subsection{Notations Required for Proofs}
	We write the probability that path $\alpha\in\mathcal{A}_T$ occurs with initial distribution $\pi \in \Delta(A)$ as
	\begin{equation}\label{eqn:Ppi hat}
		\hat{P}_\pi(\alpha)=\pi(\alpha^1)\prod_{t=1}^{T-1} \hat{P}^{\alpha^{t}}(\alpha^{t+1}),
	\end{equation}	
	where $\pi(\alpha^1)$ is the probability of $\alpha^1$ in initial distribution $\pi$.
Correspondingly, the probability that path $\alpha\in\mathcal{A}_T$ occurs with initial distribution $\pi \in \Delta(A)$ on $g\in\mathcal{G^{\mathcal{A}}}$ is
	\begin{equation}\label{eqn:Ppi}
		P_\pi(\alpha)=\pi(\alpha^1)\prod_{t=1}^{T-1} P^{\alpha^{\leq t}}(\alpha^{t+1}),
	\end{equation}	
where $\alpha^{\leq t}\in\mathcal{A}_t$ is the history $\alpha$ until time $t\in\{1,2,3,\dots,T\}$. 
We now present a result to connect the utility conditions of aligned history-dependent games to \eqref{eq:individual one step with history} and \eqref{eq:individual one step without history}.

    \begin{lemma}\label{thm:the bbP property}
    Let $P_i:A\times A\times \mathcal{U}^n\rightarrow [0,1]$ be local and increasing with respect to utility.
    Consider utility functions $\hat{U},U\in \mathcal{U}^n$.
    If  $U_i(1,a'_{-i})\geq \hat{U}_i(1,a_{-i})$, and  $\hat{U}_i(0,a_{-i})\geq U_i(0,a'_{-i})$, then $P_i(a',(1,a'_{-i}),U)\geq P_i(a,(1,a_{-i}),\hat{U})$ for some $a,a'\in A$.
    \end{lemma}
    \begin{proof}
    Assume $U_i(1,a'_{-i})\geq \hat{U}_i(1,a_{-i})$ and  $\hat{U}_i(0,a_{-i})\geq U_i(0,a'_{-i})$ for some $a,a'\in A$. Define
    \begin{equation}
    \begin{aligned}
        \bar{U}_i(x)=
        \begin{cases}
        U_i(1,a'_{-i}) & x_i=1\\
        \hat{U}_i(0,a_{-i}) & x_i=0.
        \end{cases}
    \end{aligned}
    \end{equation}
    Let $\bar{U}=(\hat{U}_1,\dots,\bar{U}_i,\dots,\hat{U}_n)$.
    Because $U(1,a'_{-i})=\hat{U}(1,a_{-i})+l$ for some $l\geq0$, we have $P_i(a,(1,a_{-i}),\bar{U})\geq P_i(a,(1,a_{-i}),\hat{U})$ by definition of increasing with respect to utility.
    Similarly, because $\hat{U}(0,a_{-i})=\bar{U}(0,a'_{-i})=U(0,a'_{-i})+l$ for some $l\geq0$, we have $P_i(a',(0,a'_{-i}),\bar{U})\geq P_i(a',(0,a'_{-i}),U)$ again by definition of increasing with respect to utility.
    By adding $P_i(a',(1,a'_{-i}),\bar{U})$ to both sides of this expression, we derive $P_i(a',(1,a'_{-i}),U)\geq P_i(a',(1,a'_{-i}),\bar{U})$.
    Using these inequalities, we find
    \begin{equation}
    \begin{aligned}
        P_i(a',(1,a'_{-i}),U)&\geq P_i(a',(1,a'_{-i}),\bar{U})\\
        &=\bar{P}_i(1,\vec{U}_i(a'_{-i})) \\
        &=\bar{P}_i(1,\vec{U}_i(a_{-i})) \\
        &=P_i(a,(1,a_{-i}),\bar{U})\\
        &\geq P_i(a,(1,a_{-i}),\hat{U}),
    \end{aligned}
    \end{equation}
    where $\bar{P}_i$ is defined according to \eqref{eq:local LR} as $P_i$ is local by hypothesis.
    \end{proof}

    Our framework requires a careful partitioning of the action space corresponding to different types of agent action deviations.
	Let $f:A\rightarrow 2^{A}$ be defined as  $$f(a)=\{a'\in A \mid a_i\neq a_i',a_{-i}=a_{-i}' \mbox{ for } i\in N\},$$ 
	which is the
	set of action profiles reachable from $a$ via exactly one unilateral deviation.
	For $a,a'\in A$, let
	\begin{equation}
			g(a,a')=
			\begin{cases}
				i & a_i\neq a_i' \\
				0 & a=a' 	
			\end{cases}
	\end{equation}
	indicate which agent unilaterally deviated its action between action profiles $a,a'$. 
	Consider $a,a'\in A$ where $a'\geq_A a$.
	We consider several disjoint subsets of $f(a)$:
	\begin{enumerate}
		\item $r(a)=\{z\in f(a) \mid a_{g(a,z)}=1\}$,
		\item $q(a,a')=\{z\in f(a)\mid  z\leq_{A} a'\}\setminus r(a)$, and
		\item $s(a,a')=f(a)\setminus (q(a,a')\cup r(a))$.
	\end{enumerate}
	These sets can be interpreted as follows. 
	The set $r(a)$ is the set of action profiles that decrease with respect to $\geq_A$; both $q(\cdot,\cdot)$ and $s(\cdot,\cdot)$ increase with respect to $\geq_A$.
	Between $q(\cdot,\cdot)$ and $s(\cdot,\cdot)$, $q(\cdot,\cdot)$'s action profiles remain less than $a'$ and $s(\cdot,\cdot)$'s profiles are greater than or incomparable to $a'$.
	We now present three more analogous sets that are disjoint subsets of $f(a')$:
	\begin{enumerate}
		\item $R(a')=\{z\in f(a') \mid a'_{g(a',z)}=0\}$,
		\item $Q(a,a')=\{z\in f(a')\mid z\geq_{A} a\}\setminus R(a')$, and
		\item $S(a,a')=f(a')\setminus (Q(a,a')\cup R(a))$.
	\end{enumerate}
	The interpretation of these sets are reversed relative to $r(\cdot)$, $q(\cdot,\cdot)$ and $s(\cdot,\cdot)$, respectively.
	
	We now highlight some useful features of these sets.
    It is evident that $r(\cdot),q(\cdot,\cdot),s(\cdot,\cdot)$ are a disjoint partition of $f(a)$, and that $R(\cdot),Q(\cdot,\cdot),S(\cdot,\cdot)$ are a disjoint partition of $f(a')$.
	For any $a,a'$ such that $a'\geq_A a$, we relate these sets by a function $b^{a,a'}:f(a)\rightarrow f(a')$.
	To evaluate $b^{a,a'}(\bar{a})$, we identify the agent that deviates its action between $a$ and $\bar{a}$ and then deviate the agent's action in $a'$. 
	Formally, we have $b^{a,a'}(\bar{a})=(\neg a'_{g(a,\bar{a})},a'_{-g(a,\bar{a})})$ where we define $\neg a_i \in \{0,1\}\setminus \{a_i\}$ for $a_i\in A_i=\{0,1\}$.
	In particular, this function relates the disjoint subsets of $f(a)$ and $f(a')$ according to the following lemma.
	\begin{lemma}\label{thm:bijection}
	    If $a,a'\in A$ and $a\leq_{A} a'$, then the following statements hold:
	    \begin{enumerate}
	        \item $b^{a,a'}:r(a)\rightarrow S(a,a')$ is a bijection,
	        \item $b^{a,a'}:s(a,a')\rightarrow R(a')$ is a bijection, and
	        \item $b^{a,a'}:q(a,a')\rightarrow Q(a,a')$ is a bijection.
	    \end{enumerate}
	\end{lemma}
\online{
\begin{proof}
    Let $a,a'\in A$ such that $a'\geq_A a$.
    We proceed by proving $b^{a,a'}:r(a)\to S(a')$ is a bijection; the other bijection statements can be proved similarly.

	We begin by proving injectiveness, namely $b^{a,a'}(z)=b^{a,a'}(z')\implies z=z'$ for $z,z'\in r(a)$.
Observe that 
$$g(a,z)=g(a',b^{a,a'}(z))=g(a',b^{a,a'}(z'))=g(a,z'),$$ where the first and third equalities follow the definition of $b^{a,a'}$ and the middle one follows the hypothesis.
Injectiveness follows from $g(a,z)=g(a,z')$, meaning that $a,z$ and $a,z'$ differ by the same agent's unilateral deviation. 
Thus, the possible actions of agent $g(a,z)$ is given by $A_{g(a,z)}\setminus \{a_{g(a,z)}\}$, which is a singleton by the binary action property, leaving only one possible state $a$ that could transition to $r(a)$ via a unilateral deviation.
	Thus $z=z'$ as desired.
	
Next we show surjection, namely that for any $z'\in S(a,a')$, there exists a $z\in r(a)$ such that $b^{a,a'}(z)=z'$ for $a,a'\in A$ and $a\leq_{A}a'$.
By definition of $S(a,a')$, we have $z'\ngeq a'$. Since $z'\in f(a')$, we observe that $z'$ and $a'$ differ only by a single unilateral deviation corresponding to some agent $i$.
By partial ordering $\leq_{A}$, we infer $a'_i=1$ and $z'_i=0$; otherwise, $z'\ngeq a'$ would be violated.
	Further, we may infer $a_i=1$; otherwise, $a_i=0$ leads to $z'\in Q(a,a')$, which gives a contradiction to the definition of $z'$.
	It is easy to see by definition of $r(a)$ that $a_i=1\implies z\in r(a)$ satisfying $g(a,z)=g(a',z')$ as $z_i\neq a_i$, but $z_{-i}=a_{-i}$ as
	$z\in f(a)$.
	Note that $g(a,z)=g(a',z')$ is always satisfied when $b^{a,a'}(z)=z'$, by the definition of the function.
\end{proof}
}
	

\subsection{The One-Step Couplings}
    To prove Theorem~\ref{thm:z functions}, we construct a monotone coupling $\mathbb{P}^{\hat{g}}_\pi$ between measures $P_\pi$ and $\hat{P}_\pi$.
    We first construct a family of monotone couplings for each one-step transition (Theorem~\ref{thm:nu}),
    and apply it to show the coupling over histories (Theorem~\ref{thm:paths}).
    
			\begin{figure*}[!htbp]
		\begin{subnumcases}{\mathbb{P}^{a,\alpha}(\bar{a},\bar{a}')=\label{eqn:nu}}
	   				\frac{1}{|N|}\left(P^\alpha_{g(\alpha^T,\bar{a}')}(1)- \hat{P}^a_{g(a',\bar{a}')}(1)\right) &  $\bar{a}=a,\bar{a}'\in R$ \label{eqn:nu:R}\\
			\frac{1}{|N|}P^\alpha_{g(\alpha^T,\bar{a}')}(\bar{a}'_{g(\alpha^T,\bar{a}')}) &  $\bar{a}=a,\bar{a}'\in Q$\label{eqn:nu:Q} \\
			\frac{1}{|N|}\left(\hat{P}^a_{g(a,\bar{a})}(0)- P^\alpha_{g(a,\bar{a})}(0)\right) &  $\bar{a}\in r,\bar{a}'=\alpha^T$ \label{eqn:nu:r} \\						
			\frac{1}{|N|}\hat{P}^a_{g(a,\bar{a})}(\bar{a}_{g(a,\bar{a})}) &  $\bar{a}\in q,\alpha^T=\bar{a}'$ \label{eqn:nu:q}\\						
			\frac{1}{|N|}\hat{P}^a_{g(a,\bar{a})}(1) & $\bar{a}=b^{a,\alpha^T}(\bar{a}'),\bar{a}'\in R$ \label{eqn:nu:sR}\\
			\frac{1}{|N|}P^\alpha_{g(\alpha^T,\bar{a}')}(0) & $\bar{a}\in r,\bar{a}'=b^{a,\alpha^T}(\bar{a})$ \label{eqn:nu:rS}\\
			\begin{array}{r}
			\frac{1}{|N|}\Big(|N|-\sum\limits_{z\in q\cup r}\hat{P}^a_{g(a,z)}(z_{g(a,z)}) 
			-\sum\limits_{z'\in Q\cup R}P^\alpha_{g(\alpha^T,z')}(z'_{g(\alpha^T,z')})\Big)
			\end{array}
			& $a=\bar{a},\alpha^T=\bar{a}'$	\label{eqn:nu:corner}		\\			
					0 & \mbox{otherwise.} \label{eqn:nu:else}
		\end{subnumcases}
		\caption{The full specification of the one-step monotone coupling for Theorem~\ref{thm:nu}.
		We adopt the notational convention that $q,s,Q,S$ are assumed to take arguments $a,a'$, and $r,R$ take arguments $a,a'$. 
		\label{fig:onestep coupling}}
		\end{figure*}
		
	\begin{theorem}
		\label{thm:nu}					
		Let $g\in \mathcal{G}^{\mathcal{A}}$ denote an aligned history-dependent game, $\hat{g}$ be its reference static game, and $P$ be a local and monotone asynchronous learning rule. 
	    Then, a monotone coupling exists between $\hat{P}^a$ and $P^\alpha$ for any $\alpha\in\mathcal{A}$ and $a\in A$ when $a\leq_A\alpha^T$. 
		This monotone coupling $\mathbb{P}^{a,\alpha}:A^2\rightarrow [0,1]$ is given in (\ref{eqn:nu}) as shown in Figure~\ref{fig:onestep coupling}.
	\end{theorem}

		\begin{proof}
			Let $a\in A$ and $\alpha\in\mathcal{A}$ such that $a\leq_A \alpha^T$. 
			Let $g\in\mathcal{G}^{\mathcal{A}}$ be an aligned history-dependent game where $\hat{g}$ is its reference static game, and let $P$ be a local and monotone asynchronous learning rule. 
			To verify that $\mathbb{P}^{a,\alpha}$ is a monotone coupling, we must show the following conditions given by  Definition~\ref{def:monotonic coupling}. For any $\bar{a},\bar{a}'\in A$, it holds that
			\begin{enumerate}
				\item \label{cond:well defined} $\mathbb{P}^{a,\alpha}$ is a well-defined probability measure, 
				\item \label{cond:lesser sum} $\sum\limits_{z'\geq_{A}\bar{a}}\mathbb{P}^{a,\alpha}(\bar{a},z')=\hat{P}^{a}(\bar{a})$, and 
				\item \label{cond:greater sum} $\sum\limits_{z\leq_{A}\bar{a}'}\mathbb{P}^{a,\alpha}(z,\bar{a}')=P^{\alpha}(\bar{a}')$.
			\end{enumerate}
We begin by verifying Condition \ref{cond:lesser sum}.
We consider cases $\bar{a}\notin (f(a)\cup \{a\}),$ $\bar{a}\in q,\bar{a}\in r,\bar{a}\in s$ and $\bar{a}=a$ separately.
We use the notational convention that $q,s,Q,S$ are assumed to take arguments $(a,\alpha^T)$, and that $r,R$ take arguments $a,\alpha^T$ respectively.
The first case represents any $\bar{a}$ which cannot be achieved in a single unilateral deviation from $a$.
This gives that $\hat{P}^a(\bar{a})=0$, and thus all pairs of $\bar{a},z'$ must satisfy $\mathbb{P}^{a,\alpha}(\bar{a},z')=0$.
This holds as all parts of \eqref{eqn:nu} require $\bar{a}\in (f(a)\cup \{a\})$ except \eqref{eqn:nu:else}, which has the desired property.
			
We now consider the second case that $\bar{a}\in q$.
Note that only (\ref{eqn:nu:q}) satisfies this condition, so we have
			\begin{equation}
				\begin{aligned}
				\sum_{z'\geq_{A}\bar{a}'}\mathbb{P}^{a,\alpha}(\bar{a},z')&=\mathbb{P}^{a,\alpha}(\bar{a},\alpha^T) \\ 
				&=\hat{P}^a_{g(a,\bar{a})}(\bar{a}_{g(a,\bar{a})})/|N| =\hat{P}^{a}(\bar{a}).
				\end{aligned}			
			\end{equation}
			
			Next we consider $\bar{a}\in r$, which satisfies \eqref{eqn:nu:r}, \eqref{eqn:nu:rS}  because $b^{a,\alpha^T}$ is a bijection by Lemma~\ref{thm:bijection}.
			Thus 
			\begin{equation}
				\begin{aligned}
					\sum_{z'\geq_{A}\bar{a}'}\mathbb{P}^{a,\alpha}(\bar{a},z')&=
					\frac{1}{|N|}\big(\hat{P}^a_{g(a,\bar{a})}(0)\\
					&\qquad-P^\alpha_{g(a,\bar{a})}(0)
					+P^\alpha_{g(\alpha^T,\bar{a}')}(0)\big)\\
					&=\frac{1}{|N|}\hat{P}^a_{g(a,\bar{a})}(0)
					=\hat{P}^{a}(\bar{a}),
				\end{aligned}
			\end{equation}
			where the second equality follows as $g(a,\bar{a})=g(\alpha^T,\bar{a}')$ by definition of $b^{a,\alpha^T}$.
			The third equality follows as $\bar{a}\in{r}\implies \bar{a}_{g(a,\bar{a})}=0$.
			
			Considering $\bar{a}\in s$, we find only \eqref{eqn:nu:sR} applies, and thus we have
			\begin{equation}
				\sum_{z'\geq_{A}\bar{a}'}\mathbb{P}^{a,\alpha}(\bar{a},z')
				=\frac{1}{|N|}\hat{P}^a_{g(a,\bar{a})}(1)
				=\hat{P}^{a}(\bar{a}),
			\end{equation}	
			where $\bar{a}\in s\implies\bar{a}_{g(a,\bar{a})}=1$ (otherwise, $\bar{a}$ would be in $q$).
			
			The final case for Condition~\ref{cond:lesser sum} is $\bar{a}=a$.
			We find cases \eqref{eqn:nu:R}, \eqref{eqn:nu:Q}, and \eqref{eqn:nu:corner} apply, yielding:
			\begin{equation}
			\begin{aligned}
				\sum_{z'\geq_{A}\bar{a}'}&\mathbb{P}^{a,\alpha}(\bar{a},z')
				=\frac{1}{|N|}\bigg(|N|-\sum_{z\in q\cup r}\hat{\bbP}^a_{g(a,z)}(z_{g(a,z)})\\
				&\qquad-\sum_{z'\in R}\hat{\bbP}^a_{g(\alpha^T,z')}(1)\bigg)\\
				&=\frac{1}{|N|}\sum_{z\in f(a)}\left(1-\hat{\bbP}^a_{g(a,z)}(z_{g(a,z)})\right)
				=\hat{P}^{a}(\bar{a}),			
				\end{aligned}
			\end{equation}
			where the first equality follows because (i) the sums over $Q\cup R$ are equivalent to the sums over $Q$ and $R$ as $Q,R$ are disjoint, and (ii) $z'\in R\Leftrightarrow z'_{g(\alpha^T,z')}=1$ by definition of $R$.
			The second equality follows as the $R$ sum is equivalent to one over $s$ by bijection $b^{a,\alpha^T}$, and thus we may combine it with the sum over $q\cup r$, leading to a sum over $f(a)$ and $|f(a)|=|N|$.
			We omit arguments for Condition~\ref{cond:greater sum} as they re in parallel to Condition~\ref{cond:lesser sum}.
			
			To verify Condition~\ref{cond:well defined}, we consider each case of (\ref{eqn:nu}) separately.
			Eqs.~(\ref{eqn:nu:Q}), (\ref{eqn:nu:q}), (\ref{eqn:nu:sR}), (\ref{eqn:nu:rS}), and (\ref{eqn:nu:else}) are trivial to verify because these probabilities are well-defined.
			Lemma~\ref{thm:the bbP property} provides the left-hand side of the equivalence:
			\begin{equation}\label{eqn:bbP equvalence}
             P^\alpha_i(1)\geq\hat{P}^a_i(1)\Leftrightarrow\hat{P}^a_i(0)\geq P^\alpha_i(0),
			\end{equation}
			where the right-hand side follows from algebraic manipulations using $P^\alpha_i(1)+P^\alpha_i(0)=1=\hat{P}^a_i(1)+P_i(0)$.
			Eq.~(\ref{eqn:nu:R}) follows directly from the hypothesis and (\ref{eqn:nu:r}) holds from the right-hand side of the equivalence.
			
			The remaining case is (\ref{eqn:nu:corner}), for which we define
			sets $N_{q}=\{g(a,z)\mid z\in q\}$, $N_{Q}=\{g(\alpha^T,z)\mid z\in Q\}$ and so on for $r,s,R,S$.
			We denote unions of these sets as $N_{qr}:=N_q\cup N_r$,  $N_{QR}:=N_Q\cup N_R$ and so on for the other combinations of $q,r,s$ and $Q,R,S$.
			Recall that $q,r,Q,R$ are partitions over states that $a,\alpha^T$ may transition to; similarly, $N_{qr}$, $N_{QR}$ are partitions of agents whose unilateral deviations lead to such transitions.
			This enables us to expand (\ref{eqn:nu:corner}) as:
            %
			\begin{equation}
			\begin{aligned} \label{eqn:nu a alpha}
				\mathbb{P}^{a,\alpha}(\bar{a},\bar{a}')&=
				\frac{1}{|N|}\Bigg(\sum_{i\in N_{qr}\cap N_{QR}}(1-\hat{P}^a_i(\neg a_i)\\
				&\qquad-P^\alpha_i(\neg \alpha^T_i))\\
				&\qquad+\sum_{i\in N_{qr}\setminus N_{QR}}(1-\hat{P}^a_i(\neg a_{i}))\\
				&\qquad+\sum_{i\in N_{QR}\setminus N_{qr}}(1-P^\alpha_i(\neg \alpha^T_{i}))\Bigg).
			\end{aligned}		
			\end{equation}
			This expansion takes advantage of $|N|=|f(a)|$ which allows $|N|$ to enter the sums as $1$.
			It now suffices to show that the summand of each sum is a well-defined probability, of which the last two terms clearly are.
			
		    We begin by investigating $i\in N_{qr}\cap N_{QR}$.
		    In particular, we have $N_q=N_Q,N_s=N_R,N_r=N_S$ due to $b^{a,\alpha^T}$ and its bijectiveness due to Lemmas~\ref{thm:bijection}.
		    By disjointness of $q,r$ we have $N_{qr}=N_{QS}$ which we apply to $N_{qr}\cap N_{QR}=N_{QS}\cap N_{QR}=N_Q=N_q$.
			Applying definitions of $q,Q$, we find $i\in N_q\implies\neg a_i=1,\neg \alpha^T_i=0$.
			Thus, the summand of the first sum for $i\in N_q$ is given by
			\begin{equation}
			    1-\hat{P}^a_i(1)-P^\alpha_i(0)\geq 1-P^\alpha_i(1)-P^\alpha_i(0)=0,
			\end{equation}
			where the inequality follows \eqref{eqn:bbP equvalence}, giving that the summands in the first term of~\eqref{eqn:nu a alpha} are themselves well-defined probabilities.
			As all conditions have been met, $\mathbb{P}^{a,\alpha}$ is a monotone coupling.
			%
		   %
		%
			%
	\end{proof}

\subsection{A Monotone Coupling Over Histories}
We now present coupling $\mathbb{P}^{\hat{g}}_\pi$, which is constructed using the one-step coupling.
    We define the indicator function $\mathds{1}$ such that $\mathds{1}(P)=1$ if $P$ is a true logical proposition and $\mathds{1}(P)=0$ otherwise.
    
\begin{theorem}\label{thm:paths}
Let $g\in\mathcal{G}^{\cal A}$ be an aligned history-dependent game, $\hat{g}$ be its corresponding static game, $P$ be a local and monotone asynchronous learning rule, and $\pi\in\Delta(\mathcal{A})$ be a distribution over all action profiles.
Then, $\mathbb{P}_\pi^{\hat{g}}:\mathcal{A}_T^2\rightarrow[0,1]$ is a monotone coupling between $\hat{P}_\pi,P_\pi$.
		This coupling is given by
		\begin{equation}\label{eqn:path coupling}
		\begin{aligned}
			\mathbb{P}^{\hat{g}}_\pi(\alpha,  \bar{\alpha})=
			\pi(\alpha^1)\mathds{1}(\alpha^1=\bar{\alpha}^1)\prod_{t=1}^{T-1}\mathbb{P}^{\alpha^t,\bar{\alpha}^{\leq t}}(\alpha^{t+1},\bar{\alpha}^{t+1}),
		\end{aligned}
		\end{equation}
where $\alpha,\bar{\alpha}\in\mathcal{A}_T$ and $\pi\in \Delta(A)$.
	\end{theorem}
	\begin{proof}
		Let $\alpha,\bar{\alpha}\in \mathcal{A}_T$ and $g\in\mathcal{G}^\mathcal{A}$. Let $\hat{g}$ be the corresponding static game and $P$ be a local and monotone asynchronous learning rule.
		We begin by showing that if $\alpha\nleq_{\mathcal{A}_T}\bar{\alpha}$, then $\mathbb{P}^{\hat{g}}_\pi(\alpha,\bar{\alpha})=0$.
		Since $\mathbb{P}^{\hat{g}}_\pi(\alpha,\bar{\alpha})=0$ if $\alpha^1\neq\bar{\alpha}^1$, we only need to consider cases where $\alpha^1=\bar{\alpha}^1$.
		Inductively, we find that if $\alpha\nleq_{\mathcal{A}_T}\bar{\alpha}$, there must exist some $t\in\{1,2,3,\dots,T-1\}$ such that $\alpha^t\leq_{A}\bar{\alpha}^t$ but $\alpha^{t+1}\nleq_{A}\bar{\alpha}^{t+1}$; let $t$ be the minimum of these values.
		We have $\mathbb{P}^{\alpha^t,\bar{\alpha}^{\leq t}}(\alpha^{t+1},\bar{\alpha}^{t+1})=0$ because $\mathbb{P}^{\alpha^t,\hat{\alpha}^{\leq t}}$ is a well-defined monotone coupling by Theorem~\ref{thm:nu}, yielding $\mathbb{P}^{\hat{g}}_\pi(\alpha,\bar{\alpha})=0$.
		It also follows that $\mathbb{P}^{\hat{g}}_\pi$ will always yield a well-defined probability as it is either $0$ or a product of well-defined probabilities.
		Thus, we only need to show that the marginal probabilities are preserved given by \eqref{eqn:monotonic coupling}. 
		We begin by showing the left-hand equation of \eqref{eqn:monotonic coupling}, namely
		\begin{equation}\label{eqn:marginal path coupling}
			\sum_{\alpha\leq_{\mathcal{A}_T}z}\mathbb{P}^{\hat{g}}_\pi(\alpha,z)=\hat{P}_\pi(\alpha) \mbox{ for each }z\in\mathcal{A}_T
		\end{equation}
		and omit the proof for the right-hand equation because it is similar.
		By inspecting (\ref{eqn:path coupling}), we observe that we only need to consider $z$ such that $z^1=\alpha^1$ and $z$ differ
		at most one unilateral deviation between any $t,t+1$.
		With these two conditions, we rewrite
		
		\begin{equation}\label{eqn:combitorial form}
		\begin{aligned}
			\sum_{\alpha\leq_{\mathcal{A}_T}z}\mathbb{P}^{\hat{g}}_\pi(\alpha,z) &=\sum_{\alpha\leq_{\mathcal{A}_T}z}\pi(\alpha^1)\prod_{t=1}^{T-1}\mathbb{P}^{\alpha^t,z^{\leq t}}(\alpha^{t+1},z^{t+1})\\
			&=\pi(\alpha^1)\sum_{\alpha^2\leq_A z^2}\hs[2]\mathbb{P}^{\alpha^1,z^{\leq 1}}(\alpha^{2},z^{2})\ldots\\
			\quad\quad&\sum_{\alpha^{T}\leq_A z^{T}}\hs[2]\mathbb{P}^{\alpha^{T-1},z^{\leq T-1}}(\alpha^{T},z^{T})
		\end{aligned}
		\end{equation}
	as the combinatorial form. 
	This allows us to apply the marginal sum properties of $\mathbb{P}^{\alpha^t,z^{\leq t}}$ from Theorem~\ref{thm:nu}, where $t\in\{1,2,..,T\}$, as follows.
	First, consider the rightmost sum in \eqref{eqn:combitorial form}. It holds that
	\begin{equation}
	\begin{aligned}
		&\sum_{\alpha^{T}\leq_A z^{T}}\mathbb{P}^{\alpha^{T-1},z^{\leq T-1}}(\alpha^{T},z^{T})=\hat{P}^{\alpha^{T-1}}(\alpha^T).
	\end{aligned}
	\end{equation}
	Because this has no dependence on $z$, we may factor out $\hat{P}^{\alpha^{T-1}}(\alpha^T)$ and repeat the process on the new rightmost sum.
	After performing this process recursively on all sums, we have 
	\begin{equation}
	\sum_{\alpha\leq_{\mathcal{A}_T}z}\mathbb{P}^{\hat{g}}_\pi(\alpha,z)=\pi(\alpha^1)\prod_{t=1}^{T-1}\hat{P}^{\alpha^{t}}(\alpha^{t+1})=\\
	\hat{P}_\pi(\alpha)
	\end{equation}
	as desired, while accounting for the indicator functions in $\mathbb{P}^{\hat{g}}_\pi$.
	This concludes the proof of Theorem~\ref{thm:paths}.
	\end{proof}

\subsection{Proof of Theorem~\ref{thm:z functions}}\label{sec:proof of z functions}

\begin{proof}
    Let $g$ be an aligned history-dependent game, $\hat{g}$ be the associated static game, $P$ be a local and monotone asynchronous learning rule, and $\pi\in\Delta(\mathcal{A})$ be a distribution over all action profiles.
    Theorem~\ref{thm:paths} gives that there exists a monotone coupling between measures $P_\pi$ and $\hat{P}_\pi$, which is given by $\mathbb{P}_\pi^{\hat{g}}$.
    This enables us to apply Proposition~\ref{thm:increasing functions in expectation} to obtain
    \begin{equation}
        \mathbb{E}_{P_\pi}(Z)-\mathbb{E}_{P_\pi}(Z)\geq 0
    \end{equation}
    for any function that is monotone with respect to ordering $\geq_{\mathcal{A}_T}$.
    This holds because the coupling $\mathbb{P}_\pi^{\hat{g}}$ is a nonnegative-valued function, which concludes the proof.
\end{proof}

\section{Conclusion}
We have developed new analytical tools to relate behaviors of history-dependent games to some related, tractable static games. This represents a significant step towards making theoretical studies more faithful to reality. 
There are many directions for future research. One is to extend the present study to more complex types of history-dependent interaction, such as games with richer action spaces and settings where the relationship between the history-dependent game and the static game holds in expectation rather than for all time. Another direction is to explore the usefulness of the concepts and techniques in other settings, such as the cybersecurity dynamics framework~\cite{XuBookChapterCD2019},
or preventive and reactive defense dynamics~\cite{XuTNSE2021-GlobalAttractivity}.
Although games have been investigated in the cybersecurity dynamics framework \cite{XuGameSec13}, much research needs to be done. Moreover, it would be interesting to investigate whether chaotic behaviors can be exhibited by history-dependence as such behaviors are relevant to proactive defense dynamics \cite{XuHotSoS2015}.

\bibliographystyle{ieeetr}
\bibliography{library}

\begin{IEEEbiography}[{\includegraphics[width=1in,height=1.25in,clip,keepaspectratio,angle=-90]{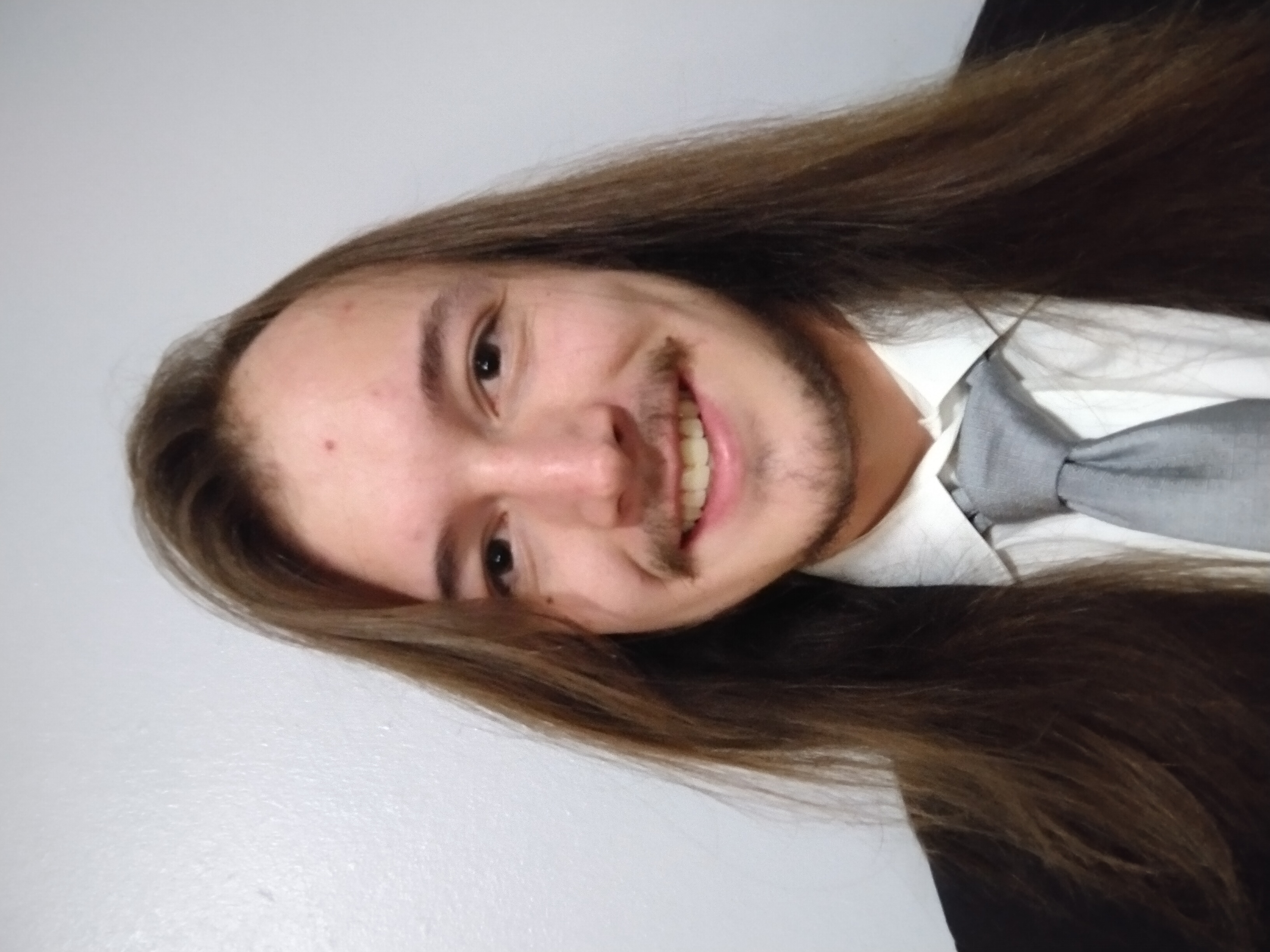}}]{Brandon C. Collins} (GS'20) 
is a Graduate Research Assistant at the University of Colorado Colorado Springs.
He received a B.S. degree in computer science from the University of Colorado Colorado Springs. He is currently pursuing a Ph.D. degree at the University of Colorado Colorado Springs.
He is interested in learning in mutli-agent systems and dynamic environments.
\end{IEEEbiography}

\begin{IEEEbiography}[{\includegraphics[width=1in,height=1.25in,clip,keepaspectratio]{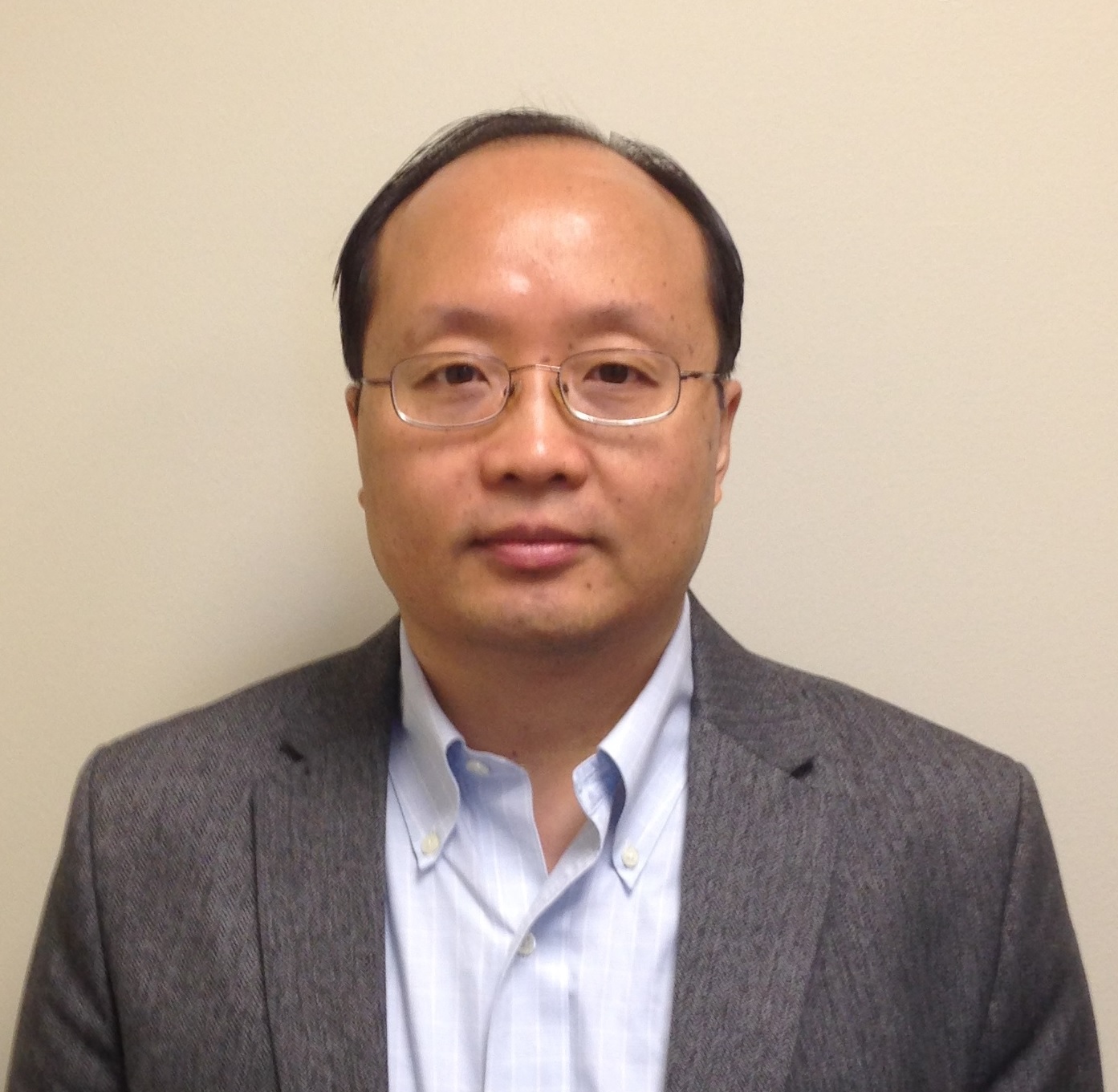}}]{Shouhuai Xu} (M’14–SM’20) received the Ph.D. degree in computer science from Fudan University
in 2000. He is the Gallogly Chair Professor in the Department of Computer Science, University of Colorado Colorado Springs (UCCS). He pioneered the Cybersecurity Dynamics approach as foundation for the emerging science of cybersecurity, with three pillars: first-principle cybersecurity modeling and analysis (the x-axis); cybersecurity data analytics (the y-axis); and cybersecurity metrics (the z-axis). He co-initiated the International Conference on Science
of Cyber Security and is serving as its Steering Committee Chair. He is/was an Associate Editor of IEEE Transactions on Dependable and Secure Computing (IEEE TDSC), IEEE Transactions on Information Forensics and Security (IEEE T-IFS), and IEEE Transactions on Network
Science and Engineering (IEEE TNSE).
\end{IEEEbiography}
\vspace{-10pt}

\begin{IEEEbiography}[{\includegraphics[width=1in,height=1.25in,clip,keepaspectratio]{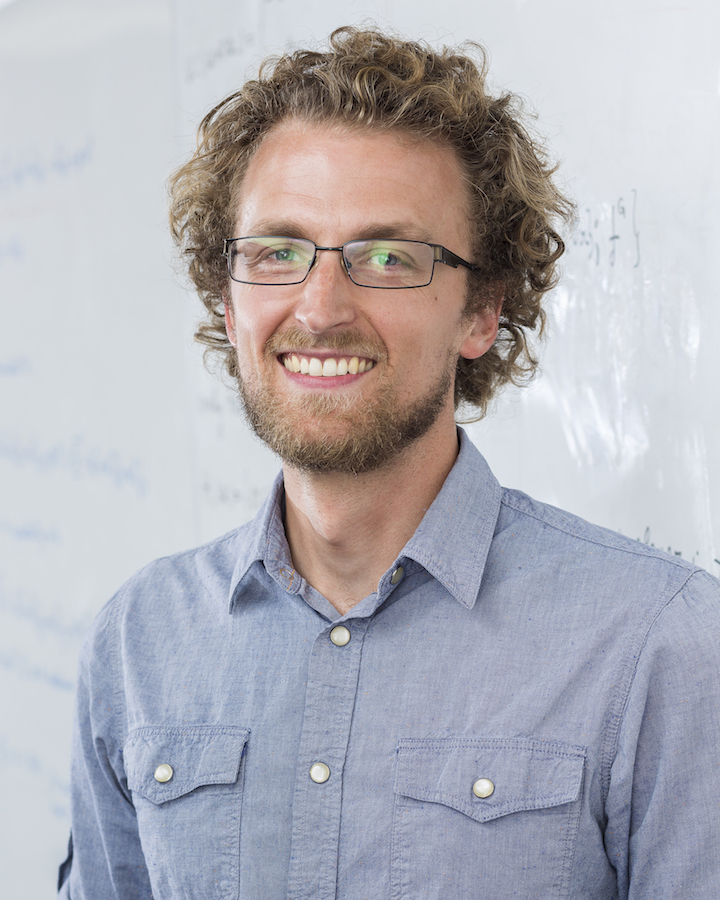}}]{Philip N. Brown}
 is an Assistant Professor in the Department of Computer Science at the University of Colorado Colorado Springs. Philip received the Bachelor of Science in Electrical Engineering in 2007 from Georgia Tech, after which he spent several years designing control systems and process technology for the biodiesel industry. He received the Master of Science in Electrical Engineering in 2015 from the University of Colorado at Boulder under the supervision of Jason R. Marden, where he was a recipient of the University of Colorado Chancellor's Fellowship. He received the PhD in Electrical and Computer Engineering from the University of California, Santa Barbara under the supervision of Jason R. Marden. He received the 2018 CCDC Best PhD Thesis Award from UCSB and the Best Paper Award from GameNets 2021. Philip is interested in the interactions between engineered and social systems.
\end{IEEEbiography}

\end{document}